\documentclass[11pt]{amsart}
\usepackage{geometry}                
\usepackage[latin1]{inputenc}   
\usepackage{mathpazo}  
\usepackage[english]{babel}
\usepackage[usenames,dvipsnames,cmyk]{xcolor}
\usepackage{url}
\usepackage{amsmath,amssymb,latexsym,mathrsfs,amssymb,amsthm}
\usepackage{enumerate}
\usepackage{mathtools} 
\usepackage[all]{xy}
\usepackage{setspace}
\usepackage{tikz, tkz-euclide}
\usepackage[colorlinks=true, linkcolor=MidnightBlue, citecolor=MidnightBlue, urlcolor=MidnightBlue]{hyperref}

\newtheorem{theorem}{Theorem}[section]
\newtheorem{lemma}[theorem]{Lemma}
\newtheorem{proposition}[theorem]{Proposition}

\newtheorem{assumption}{Assumption}

\theoremstyle{definition}

\theoremstyle{remark}
\newtheorem{remark}[theorem]{Remark}
\newtheorem{rhp}{Riemann-Hilbert Problem}
\newcommand{\rhref}[1]{Riemann-Hilbert Problem~\ref{#1}}

\let\Im=\undefined\DeclareMathOperator{\Im}{Im}

\newcommand{\dd}{\ensuremath{\,\mathrm{d}}}
\newcommand{\ii}{\ensuremath{\mathrm{i}}}
\newcommand{\ee}{\ensuremath{\,\mathrm{e}}}

\newcommand{\dummy}{\mathbf{H}}

\makeatletter
\renewcommand*\env@matrix[1][\arraystretch]{%
  \edef\arraystretch{#1}%
  \hskip -\arraycolsep
  \let\@ifnextchar\new@ifnextchar
  \array{*\c@MaxMatrixCols c}}
\makeatother

\let\originalleft\left
\let\originalright\right
\renewcommand{\left}{\mathopen{}\mathclose\bgroup\originalleft}
\renewcommand{\right}{\aftergroup\egroup\originalright}

\title{Broader Universality of Rogue Waves of Infinite Order}

\author{Deniz Bilman}
\address{Deniz Bilman:  Department of Mathematical Sciences, University of Cincinnati, Cincinnati, OH, USA}
\email{bilman@uc.edu}

\author{Peter D.~Miller}
\address{Peter D. Miller:  Department of Mathematics, University of Michigan, Ann Arbor, MI, USA}
\email{millerpd@umich.edu}

\thanks{Bilman was supported by the National Science Foundation under grant number DMS-2108029.
Miller was supported by the National Science Foundation under grant number DMS-1812625.
}

\keywords{Rogue waves;  Inverse scattering transform;  Nonlinear Schr\"odinger equation.}
\date{\today}

\begin{document}
\begin{abstract}
We show that the same special solution of the focusing nonlinear Schr\"odinger equation that has been shown to arise in a certain near-field/large-order limit from soliton and Peregrine-like rogue wave solutions actually arises universally from an arbitrary background solution when subjected to a sequence of iterated B\"acklund transformations.
\end{abstract}
\maketitle
\section{Introduction}
This paper concerns \emph{extreme superposition} of soliton-like 
solutions of the focusing nonlinear Schr\"odinger (NLS) equation
\begin{equation}
\ii q_t +\frac{1}{2}q_{xx} + |q|^2 q = 0, \qquad(x,t)\in\mathbb{R}^2,
\label{nls}
\end{equation}
on a class of arbitrary background fields given in terms of representations in the form of a Riemann-Hilbert problem. The focusing NLS equation in the form \eqref{nls} is the $\lambda$-independent compatibility condition for the simultaneous linear equations of a Lax pair
\begin{equation}
\mathbf{u}_{x} = \begin{bmatrix}-\ii \lambda & q \\ -q^* & \ii \lambda\end{bmatrix} \mathbf{u},\qquad
\mathbf{u}_{t} = \begin{bmatrix}-\ii  \lambda^{2}+\ii  \frac{1}{2}|q|^{2} & \lambda q+\ii  \frac{1}{2}  q_{x} \\ -\lambda  q^{*}+\ii  \frac{1}{2}  q_{x}^{*} & \ii  \lambda^{2}-\ii  \frac{1}{2}| q|^{2} \end{bmatrix}\mathbf{u},\label{Lax-pair}
\end{equation}
governing an auxiliary column vector $\mathbf{u}=\mathbf{u}(\lambda;x,t)$ and the spectral parameter $\lambda\in\mathbb{C}$.

{
The NLS equation \eqref{nls} is well-known to be a universal amplitude equation for weakly-nonlinear and nearly monochromatic wavetrains in dispersive conservative systems \cite{BenneyN1967}.  In recent years, there has been increasing interest in \emph{rogue waves} in various physical systems.  These waves are informally characterized as having a relatively large amplitude in a space-time localized region.   The exact solution 
\begin{equation}
q(x,t)=\ee^{\ii t}\left[1-4\frac{1+2\ii t}{1+4x^2+4t^2}\right]
\end{equation}
of \eqref{nls} found in 1983 by D. H. Peregrine \cite{Peregrine1983} has this property, reaching a peak amplitude of $|q(0,0)|=3$ and decaying in all space-time directions to the background amplitude of $|q_0(x,t)|=1$.  There is now a vast body of literature devoted to generalizing this exact solution to a family of ``higher-order'' versions thereof (see, for instance, \cite{GuoTYLW2017} and references therein).  These solutions all decay for large $(x,t)$ to the background solution $q_0(x,t)=\ee^{\ii t}$ (representing a uniform train of Stokes waves in the underlying physical system for which \eqref{nls} is the amplitude equation) and they have an algebraic character, being derived from iterated Darboux transformations or Hirota's bilinear method, both of which lead to determinantal formul\ae.  At the same time, many researchers have also investigated other types of solutions that can exhibit large space-time localized bursts of amplitude without necessarily decaying to a constant-amplitude background state for large $(x,t)$.  An example of such a study in the context of the modified Korteweg-de Vries (mKdV) equation is \cite{SlunyaevP2016}; in this paper the authors consider colliding solitons as a mechanism for the generation of rogue waves near the location and instant of collision. The idea is that, by tuning the parameters of the colliding solitons, one can maximize the amplitude burst.  It turns out that the largest amplitude burst is generated by collisions of solitons of alternating signs (mKdV solitons are characterized by an arbitrary sign).  Another related work is \cite{KevrekidisLS2020}, in which the approach taken is to modify the model equation so that a Peregrine-like exact solution of fourth order can achieve an amplitude burst of nearly a thousand times the background level with appropriate choice of its parameters.  Such a large-amplitude burst is made possible in this context because the model equation involves two coupled fields and has a sign-indefinite conserved ``norm'' which allows both fields to become large while the norm is conserved.
}

{
Another approach to finding solutions of \eqref{nls} modeling rogue waves with very large amplitudes is to consider exact solutions of increasingly-high order.  It has been shown that if the parameters of the $k^\mathrm{th}$-order Peregrine-type solution $q=q_k(x,t)$ of \eqref{nls} are chosen to concentrate the solution near a chosen focal point $(x_0,t_0)$, then the maximum amplitude $|q_k(x_0,t_0)|$ is equal to $2k+1$, and hence can be made arbitrarily large (see, for instance, \cite{WangYWH2017}).  In \cite{BilmanLM2020}, this so-called ``fundamental'' rogue wave solution was investigated in the \emph{near-field limit} in which the space-time coordinates are scaled near the focal point $(x_0,t_0)$ as the order increases.  Not only does this analysis reproduce the previously-known maximal amplitude result, but more significantly it reveals that in the near-field/high-order limit the scaled field $k^{-1}q_k(x,t)$ converges as $k\to\infty$ to a limiting solution $Q(X,T)$ of \eqref{nls} written in the rescaled coordinates $(X,T)$ termed the \emph{rogue wave of infinite order}.  It is a solution decaying to zero in all directions of space-time (although quite slowly; $Q(X,T)=O(X^{-3/4})$ as $X\to\pm\infty$ for fixed $T$, and $Q(X,T)=O(T^{-1/3})$ as $T\to\pm\infty$ for fixed $X$) and hence it may be viewed as a \emph{rogue wave on the zero background}.  Just as large amplitudes can be achieved in solutions of \eqref{nls} by mechanisms other than those modeled by Peregrine-type solutions, the same solution $Q(X,T)$ can also be generated by mechanisms like soliton interactions.  Indeed, in \cite{BilmanB2019} a near-field scaling limit of soliton solutions of \eqref{nls} (satisfying zero boundary conditions $q(x,t)\to 0$ as $x\to\pm\infty$) corresponding to a single complex-conjugate pair of high-order eigenvalues of the Zakharov-Shabat operator reveals a one-parameter generalization of $Q(X,T)$ as a limiting profile.  In \cite{BilmanM2021}, the high-order solitons and Peregrine-type rogue waves were embedded into a single family of solutions of \eqref{nls} with a continuous parameter $M$ representing the order such that for $M=\tfrac{1}{2}k$ the solution is a soliton solution of order $k$ while for $M=\tfrac{1}{4}+\tfrac{1}{2}k$ it is instead a Peregrine-type fundamental rogue wave solution of order $k$.  Then it was shown that the same near-field limit solution $Q(X,T)$ arises from this solution family in the limit $M\to\infty$ even through sequences for which the solution is neither an exact soliton solution nor an exact Peregrine-type solution.  Moreover, in the same paper the universal asymptotic behavior in the limit $M\to\infty$ was extended to a \emph{far-field regime} in which the independent variables $(x,t)$ are allowed to become large in such a way that rescaled coordinates $(\chi,\tau):=(x/M,t/M)$ lie within a definite bounded domain.
}

{
Both the high-order solitons and the high-order Peregrine-type solutions of \eqref{nls} arise from iterated B\"acklund transformations applied to a background solution $q_0(x,t)$, also called a ``seed'' ($q_0(x,t)\equiv 0$ for solitons, and $q_0(x,t)\equiv \ee^{\ii t}$ for Peregrine-type solutions), and for both types of solutions, the near-field asymptotic behavior is described by the same function $Q(X,T)$ (or a one-parameter generalization thereof).  In this paper, we demonstrate that this universality is by no means restricted to elementary background solutions.  Indeed, we show that for a wide class of background solutions $q_0(x,t)$ serving as a seed for iterated B\"acklund transformations, the same rogue wave of infinite order arises in the near-field/high-order limit.  This shows that, in a sense, \emph{all rogue wave solutions of \eqref{nls} having sufficiently large amplitude look the same near the focal point $(x_0,t_0)$}.  
}

{
In Section~\ref{sec:background}, we describe a class of background fields $q=q_0(x,t)$ satisfying \eqref{nls} that will serve as seeds for iterated B\"acklund transformations. Then, in Section~\ref{sec:superposition}, we implement the iterated B\"acklund transformations (and the coincident Darboux transformations of associated eigenfunctions of the Zakharov-Shabat problem), and show how to pass to the near-field/high-order limit, proving rigorously that the same solution $Q(X,T)$ characterizes the limit regardless of the background field.  Our main result is then formulated as Theorem~\ref{thm:main}.
}

\subsubsection*{On notation} 
Boldface capital letters denote matrices and boldface lowercase letters denote vectors.
We use $a^*$ to denote the complex conjugate of a quantity $a\in\mathbb{C}$ and for a matrix $\mathbf{A}$, $\mathbf{A}^*$ denotes the component-wise complex conjugate without the transpose. We use $\mathbf{A}^\dagger$ to denote the conjugate transpose of a matrix $\mathbf{A}$. For a set $S\subset{\mathbb{C}}$, we use $|S|$ to denote the set of pointwise moduli $|S|:=\{ |z| \colon z\in S\}$; for a curve $L\in\mathbb{R}^2$ we use $|L|$ to denote the arc length of $L$. It will be clear from context what $| \diamond |$ stands for. {We denote the standard Pauli spin matrices by
\begin{equation}
\sigma_1:=\begin{bmatrix}0&1\\1&0\end{bmatrix},\quad\sigma_2:=\begin{bmatrix}0 & -\ii\\\ii & 0\end{bmatrix},\quad
\sigma_3:=\begin{bmatrix}1&0\\0 & -1\end{bmatrix}.
\end{equation}
For a matrix-valued function $\lambda\mapsto\mathbf{M}(\lambda)$ analytic on the complement of an oriented arc $\Sigma$ and a given point $\zeta\in\Sigma$ we denote the boundary value taken by $\mathbf{M}(\lambda)$ as $\lambda\to\zeta$ from the left (resp., right) of $\Sigma$ as $\mathbf{M}_+(\zeta)$ (resp., $\mathbf{M}_-(\zeta)$).
}

\section{Background Fields}
\label{sec:background}
We make the following assumptions for a contour $\Sigma$ in the complex plane and a jump matrix $\mathbf{V}_0(\zeta)$ supported on $\Sigma$.  
We let $\|\diamond\|$ denote an arbitrary vector norm on $\mathbb{C}^2$ and we use the same notation for the induced norm on $2\times 2$ matrices over $\mathbb{C}$.

\begin{assumption}[Jump contour]
\label{a:jump-contour}
We assume that $\Sigma$ is a Schwarz-symmetric oriented contour in $\mathbb{C}$ with the following properties
\begin{itemize}
\item $\mathbb{R}\subseteq {\Sigma}$ and this is the only unbounded component of $\Sigma$, which we orient from $-\infty$ to $+\infty$.
\item $\Sigma \cap \mathbb{C}^+$ consists of finitely many pairwise-disjoint Jordan curves ${\Sigma}_j^+$, $j=1,2,\ldots,N$, which we refer to as \emph{loops}, and each loop $\Sigma_j^+$ is oriented clockwise.
\item Due to the aforementioned Schwarz symmetry, $\Sigma \cap \mathbb{C}^-$ consists of loops ${\Sigma}_j^-$ that are oriented counter-clockwise and that satisfy $({\Sigma}_j^-)^*={\Sigma}_j^+$ as subsets of the complex plane.
\end{itemize}
\end{assumption}

Note that $\Sigma$ is a complete oriented contour in the sense that it divides the complex plane into two complementary regions, $\mathbb{C}=\Omega^+\sqcup\Sigma\sqcup\Omega^-$, where $\Omega^+$ (resp.\@ $\Omega^-$) denotes the union of the regions that lie to the left (resp.\@ right) side of $\Sigma$ with respect to its orientation.

\begin{assumption}[Jump matrix] 
\label{a:jump-matrix}
We assume that $\mathbf{V}_0(\zeta)$ is a $2\times 2$ complex matrix valued function defined for all $\zeta\in\Sigma$ satisfying for some constants $K>0$, $K'>0$, and $0<\nu < 1$ the following properties:
\begin{itemize}
\item \emph{Unimodularity:} $\det(\mathbf{V}_0(\zeta))\equiv 1$ for all $\zeta\in\Sigma$.
\item \emph{H\"older continuity/decay to identity:} $\| \mathbf{V}_0(\zeta_1) - \mathbf{V}_0(\zeta_2) \| \leq K |\zeta_1 - \zeta_2|^\nu$, for any $\zeta_1,\zeta_2$ belonging to the same smooth component of $\Sigma$ (i.e., to a loop or $\mathbb{R}$).
There exists $\varepsilon\ge \nu$ such that for $\zeta_1, \zeta_2 \in \mathbb{R}$ with $|\zeta_1|\geq1$ and $|\zeta_2|\geq 1$,
\begin{equation}
\| \zeta_1^{4+\varepsilon}(\mathbf{V}_0(\zeta_1)-\mathbb{I}) - \zeta_2^{4+\varepsilon}(\mathbf{V}_0(\zeta_2)-\mathbb{I})\|  \leq K' \left|\frac{1}{\zeta_1} - \frac{1}{\zeta_2}\right|^\nu.
\label{Holder-infinity}
\end{equation}

\item \emph{Schwarz reflection symmetry:} For all $\zeta\in\mathbb{R}$, $\mathbf{V}_0(\zeta) +\mathbf{V}_0(\zeta)^\dagger$ is strictly positive definite. For all $\zeta \in \Sigma \setminus \mathbb{R}$, $\mathbf{V}_0(\zeta^*)^\dagger = \mathbf{V}_0(\zeta)$.
\end{itemize}
\end{assumption}
Note that in particular we have $\|\mathbf{V}_0(\zeta)-\mathbb{I}\|=O(|\zeta|^{-4-\varepsilon-\nu})$ which is dominated by $O(|\zeta|^{-4-2\nu})$ as $|\zeta|\to\infty$.
We denote by $\mathbf{V}_0^\mathbb{R}(\zeta)$ the restriction of $\mathbf{V}_0(\zeta)$ on $\mathbb{R}$,
fix $\nu\in(0,1)$ to be the H\"older exponent of the jump matrix $\mathbf{V}_0(\zeta)$ and consider the following Riemann-Hilbert problem.
\begin{rhp}[Background Field]\label{rhp:background}
Find a $2\times 2$ matrix valued function $\mathbf{M}(\lambda;x,t)$ with the following properties.
\begin{itemize}
\item {\bf Analyticity:} $\mathbf{M}(\lambda;x,t)$ is analytic in $\lambda$ for $\lambda \in \mathbb{C}\setminus \Sigma$.
\item {\bf Jump conditions:} $\mathbf{M}(\lambda;x,t)$ admits H\"older continuous boundary values on $\Sigma$ for all exponents $\mu<\nu$ as $\lambda\to \zeta$ for $\zeta\in\Sigma$, and these boundary values are related by the following jump condition
\begin{equation}
\mathbf{M}_+(\zeta;x,t)=\mathbf{M}_-(\zeta;x,t) \ee^{-\ii \zeta (x + \zeta t)\sigma_3} \mathbf{V}_0(\zeta) \ee^{\ii \zeta (x + \zeta t)\sigma_3},\qquad \zeta \in \Sigma,
\end{equation}
where the boundary values
\begin{equation*}
\mathbf{M}_{\pm}(\zeta;x,t)  :=\lim_{\substack{{\lambda\to\zeta}\\{\lambda\in\Omega^\pm}}} \mathbf{M}(\lambda;x,t)
\end{equation*}
exist independently of the path of approach.
\item {\bf Normalization:} $\mathbf{M}(\lambda;x,t) \to \mathbb{I}$ as $\lambda\to\infty$ in $\mathbb{C}\setminus \Sigma$.
\end{itemize}
\end{rhp}

The jump contour and the jump matrix associated with \rhref{rhp:background} satisfy all the hypotheses of Zhou's \emph{vanishing lemma} \cite[Theorem 9.3]{Zhou89} thanks to the properties listed in Assumption~\ref{a:jump-contour} and Assumption~\ref{a:jump-matrix}. Therefore, \rhref{rhp:background} has a unique solution $\mathbf{M}(\lambda)= \mathbf{M}(\lambda;x,t)$ for any $(x,t)\in\mathbb{R}^2$, and necessarily, $\det(\mathbf{M}(\lambda;x,t))\equiv 1$. The decay rate assumption on $\mathbf{V}_0(\zeta)-\mathbb{I}$ as $\zeta\to\infty$ in Assumption~\ref{a:jump-matrix} also provides differentiability of $\mathbf{M}(\lambda;x,t)$ with respect to $x$ and $t$ since $\zeta^2(\mathbf{V}_0(\zeta)-\mathbb{I}) = O(|\zeta|^{-2-2\nu})$ as $|\zeta|\to \infty$ on $\mathbb{R}$ (see, for example, \cite[Lemma 3.8]{TrogdonO-book}). Thus, it follows by a \emph{dressing} calculation (see, for instance, \cite[Section 3.2]{BilmanM2019}) that \rhref{rhp:background} defines a global solution $q=q_0(x,t)$ of the focusing NLS equation \eqref{nls} defined by the residue at $\lambda=\infty$:
\begin{equation}
q_0(x,t):= 2\ii  \lim_{\lambda\to\infty} \lambda M_{12}(\lambda;x,t).
\label{q0-def}
\end{equation}

We call $q_0(x,t)$ the background field in the context of this paper.  The conditions of \rhref{rhp:background} are sufficiently general for the set of background fields $q_0(x,t)$ to include:
\begin{itemize}
\item Generic solutions of the Cauchy initial-value problem for \eqref{nls} with zero boundary conditions at infinity (say with initial data in the Schwartz space $\mathscr{S}(\mathbb{R})$).  Here by generic, we mean to exclude only real spectral singularities, i.e., real zeros of the scattering coefficient $a(\lambda)$ (reciprocal of the transmission coefficient).  However, initial data generating finitely-many complex eigenvalues of arbitrary algebraic multiplicities are taken into account via the jump conditions on the various loops $\Sigma_j^\pm$; one may start with the Beals-Coifman solution of the Zakharov-Shabat problem which has a pole at each eigenvalue with order equal to the algebraic multiplicity and then introduce a local substitution within the corresponding loop based on the discrete spectral data for the pole (in the case of a simple pole, this is just the connection coefficient or norming constant) to remove it at the cost of a jump on the loop.
\item Certain analogues of \emph{primitive potentials} \cite{DyachenkoZZ16}. These are solutions (so far mostly studied for KdV-type equations) that are designed as models for soliton turbulence; they are constructed through a limiting process of inserting increasingly many eigenvalues that accumulate in the limit along a curve in the upper half-plane, along with its image in the lower half-plane.  The limiting process results in a Schwarz-symmetric jump condition across one or more arcs (``condensation'' of the eigenvalue poles), possibly in addition to a jump on the real line from a nonzero reflection coefficient (usually omitted).  In cases where the resulting jumps on all arcs tend to the identity at the arc endpoints, we may fit this into the framework of \rhref{rhp:background} by ``completing'' each arc to form a loop $\Sigma_j^\pm$ by taking the jump matrix to be the identity on the complement of the arc in the loop.
\end{itemize}
{Ultimately, we may also consider arbitrary solutions of \eqref{nls} satisfying nonzero boundary conditions at infinity as background fields, although for that purpose it is easier to start with a different Riemann-Hilbert formulation; see Remark~\ref{rem:nonzero-background} below.}

There are two properties of the solution of \rhref{rhp:background} that are essential for our purposes given in the following two propositions.

\begin{proposition}
For each $(x_0,t_0)\in\mathbb{R}^2$, the solution of \rhref{rhp:background} satisfies $\mathbf{M}(\lambda;x_0,t_0) \to \mathbb{I}$ as $\lambda\to\infty$ in $\mathbb{C}\setminus \Sigma$ uniformly in all directions.
\label{prop:independence}
\end{proposition}

\begin{proposition}
Fix $(x_0,t_0)\in\mathbb{R}^2$.
Then $\sup_{\lambda\in\mathbb{C}\setminus\Sigma} \|\mathbf{M}(\lambda;x,t)-\mathbf{M}(\lambda;x_0,t_0)\|\to 0$ as $(x,t)\to(x_0,t_0)$ 
in $\mathbb{R}^2$.
\label{prop:continuity}
\end{proposition}

To prove these propositions, we first map the unbounded contour $\Sigma$ to a compact one with the help of the M\"obius transformation from the $\lambda$-plane to the $z$-plane defined by
\begin{equation}
z=z(\lambda):= \frac{\lambda - \ii r}{\lambda+\ii r},\qquad \lambda=\lambda(z):= -\ii r \frac{z+1}{z-1},
\label{Mobius}
\end{equation}
where $r>0$ is any number satisfying $r> \sup |\Sigma\setminus \mathbb{R}|$ so that the loops $\Sigma_j^{\pm}$, $j=1,2,\dots,N$, all lie within the disk $|\lambda|< r$. The fractional linear transformation $\lambda\mapsto z$ maps the upper (resp.\@ lower) half $\lambda$-plane to the interior (resp.\@ exterior) of the unit circle $\mathbb{T}$ in the $z$-plane and the image $z(\mathbb{R}) = \mathbb{T}$ of $\mathbb{R}$ is oriented counter-clockwise. Moreover, it maps the loop $\Sigma_j^+$ (resp.\@ $\Sigma_j^-$) to a loop $\Gamma_j^{+}$ (resp.\@ $\Gamma_j^{-}$) in the interior (resp.\@ exterior) of the unit circle $\mathbb{T}$ in the $z$-plane. Importantly, no point of the jump contour $\Sigma$ is mapped to $z=\infty$. The contour system $\Gamma:= z(\Sigma) = \bigcup_{j=1}^N (\Gamma_j^{+} \cup \Gamma_j^{-}) \cup \mathbb{T}$ is compact, consisting of a union of finitely many disjoint loops and hence does not have any self-intersection points. The orientations of the loops are preserved: the images $z(\Sigma^+_j) = \Gamma_j^{+}$ are oriented clockwise and $z(\Sigma^-_j) = \Gamma_j^{-}$ are oriented counter-clockwise, $j=1,2,\ldots,N$. Finally, note that the point $\lambda=\infty$ is mapped to $z=1$, and $z=\infty$ is the image of $\lambda=-\ii r$. 

We set $\mathbf{W}_0(z):= \mathbf{V}_0(\lambda(z))$, for $z \in \Gamma$. It is clear from the definition that $\det(\mathbf{W}_0(z))=1$. The H\"older continuity of $\mathbf{V}_0(\zeta)$ on $\Sigma$ and the H\"older condition \eqref{Holder-infinity} at infinity are sufficient to guarantee that $\mathbf{W}_0(z)$ satisfies the H\"older condition on $\Gamma$ with exponent $\nu$ including at the point $z=1$,
the image of $\lambda=\infty$ under the mapping \eqref{Mobius}.

We consider the following Riemann-Hilbert problem formulated in the $z$-plane (see \eqref{Mobius}) for the compact jump contour $\Gamma=z(\Sigma)$ and the ``core'' jump matrix $\mathbf{W}_0(z)$.
\begin{rhp}
\label{rhp:compact}
Find a $2\times 2$ matrix valued function $\mathbf{N}(z;x,t)$ with the following properties.
\begin{itemize}
\item {\bf Analyticity:} $\mathbf{N}(z;x,t)$ is analytic in $z$ for $z \in \mathbb{C}\setminus \Gamma$.
\item {\bf Jump conditions:} $\mathbf{N}(z;x,t)$ admits H\"older continuous boundary values on $\Gamma$ for all exponents $\mu<\nu$ as $z\to s$ for $s\in\Sigma$, and these boundary values are related by the following jump condition
\begin{equation}
\mathbf{N}_+(s;x,t)=\mathbf{N}_-(s;x,t) \ee^{-\ii \lambda(s) (x + \lambda(s) t)\sigma_3} \mathbf{W}_0(s) \ee^{\ii \lambda(s) (x + \lambda(s) t)\sigma_3},\qquad s \in \Gamma,
\label{N-jump}
\end{equation}
where the boundary values
\begin{equation*}
\mathbf{N}_{\pm}(s;x,t)  :=\lim_{\substack{{z\to s}\\{z\in z(\Omega^\pm)}}} \mathbf{N}(z;x,t)
\end{equation*}
exist independently of the path of approach.
\item {\bf Normalization:} $\mathbf{N}(z;x,t) \to \mathbb{I}$ as $z\to\infty$ in $\mathbb{C}\setminus \Gamma$.
\end{itemize}
\end{rhp}
The correspondence between solutions of the two Riemann-Hilbert problems stated above is the following. Given the solution $\mathbf{M}(z;x,t)$ of \rhref{rhp:background}, the matrix $\mathbf{N}(z;x,t):=\mathbf{M}(-\ii r; x, t)^{-1} \mathbf{M}(\lambda(z);x,t)$ is a solution of \rhref{rhp:compact} with the properties stated therein (see the proof of \cite[Lemma A.2.1]{KamvissisMM2003} for details), and it is easy to show using Liouville's theorem that this solution is unique, satisfying $\det(\mathbf{N}(z;x,t))\equiv 1$. Conversely, given a solution $\mathbf{N}(z;x,t)$ of \rhref{rhp:compact}, $\mathbf{M}(\lambda;x,t):=\mathbf{N}(1;x,t)^{-1} \mathbf{N}(z(\lambda);x,t)$ defines a solution of \rhref{rhp:background} with all of the properties listed therein. 
To prove Proposition~\ref{prop:independence} and Proposition~\ref{prop:continuity},
we make use of the bijective correspondence between the solutions of these Riemann-Hilbert problems and rely on Fredholm theory for Riemann-Hilbert problems in H\"older spaces presented in \cite[Section A.4 and Section A.5]{KamvissisMM2003}.

We define the Cauchy integral of a matrix-valued function $\dummy$ defined on $\Gamma$ by
\begin{equation}
\mathcal{C}^\Gamma[\dummy](z) := \frac{1}{2\pi \ii}\int_\Gamma \frac{\dummy(s)}{s-z}\dd s,\quad z\in\mathbb{C}\setminus\Gamma.
\end{equation}
We encode the parametric dependence on $(x,t)$ of the jump matrix in \eqref{N-jump} by introducing
\begin{equation}
\overline{\mathbf{W}}_0(s;x,t):= \ee^{-\ii \lambda(s) ( x + \lambda(s) t) \sigma_3} \mathbf{W}_0(s)  \ee^{\ii \lambda(s)( x + \lambda(s) t) \sigma_3}, \qquad s\in \Gamma,
\end{equation}
and we let $\mathcal{T}(x,t)$ denote the Beals-Coifman integral operator acting on matrix-valued functions defined on $\Gamma$ and depending parametrically on $(x,t)\in\mathbb{R}^2$:
\begin{equation}
\mathcal{T}(x,t)[\dummy](s):= \mathcal{C}_-^\Gamma[\dummy(\diamond)(\overline{\mathbf{W}}_0(\diamond;x,t)-\mathbb{I})](s)=\frac{1}{2\pi\ii}\int_\Gamma\frac{\dummy(w)(\overline{\mathbf{W}}_0(w;x,t) - \mathbb{I})}{w-s_-} \dd w, \qquad s\in\Gamma.
\end{equation}
Here the ``$-$'' subscript indicates taking a boundary value at  $s\in\Gamma$ from $z(\Omega^-)$.

We denote by $C^\nu(\Gamma)$ the Banach algebra of $2\times 2$ complex-valued matrix functions $\dummy$ that are H\"older continuous on $\Gamma$ with exponent $\nu\in(0,1)$, and we let $\| \dummy \|_{C^\nu(\Gamma)}$ denote the norm of $\dummy\in C^\nu(\Gamma)$ defined by
\begin{equation}
\| \dummy \|_{C^\nu(\Gamma)} := \| \dummy \|_{L^\infty(\Gamma)} + \sup_{\substack{s_1,s_2\in\Gamma \\ s_2 \neq s_1}} \frac{\| \dummy(s_2) - \dummy(s_1)\|}{|s_2 - s_1|^\nu}.
\end{equation}
We denote by $\| \diamond \|_{C^\nu(\Gamma) \circlearrowleft}$ the induced operator norm on the algebra of bounded operators on $C^\nu(\Gamma)$.
\begin{lemma}
For each $(x,t)\in\mathbb{R}^2$, $\overline{\mathbf{W}}_0(\diamond;x,t)-\mathbb{I}\in C^\nu(\Gamma)$.
\label{lem:HoelderJumpGamma}
\end{lemma}
\begin{proof}
Let $D$ denote the disk centered at $z=1$ of small radius $\delta>0$.  It suffices to prove that $\overline{\mathbf{W}}_0(\diamond;x,t)-\mathbb{I}\in C^\nu(\Gamma\cap D)\cap C^\nu(\Gamma\setminus D)$.

To prove that $\overline{\mathbf{W}}_0(\diamond;x,t)-\mathbb{I}\in C^\nu(\Gamma\setminus D)$,  we note that the constant function $\mathbb{I}\in C^\nu(\Gamma\setminus D)$, and the conjugating factors $\ee^{\pm\ii\lambda(s)(x+\lambda(s)t)\sigma_3}$ are analytic away from $s=1$ and hence also elements of the algebra $C^\nu(\Gamma\setminus D)$.  By conformality of $z\mapsto \lambda(z)$ away from $z=1$ and the condition that $\|\mathbf{V}_0(\zeta_2)-\mathbf{V}_0(\zeta_1)\|\le K|\zeta_2-\zeta_1|^\nu$ for bounded $\zeta_1,\zeta_2\in\Sigma$ it then follows that $\mathbf{W}_0(\diamond)\in C^\nu(\Gamma\setminus D)$.
Restoring the conjugating factors, $\overline{\mathbf{W}}_0(\diamond;x,t)\in C^\nu(\Gamma\setminus D)$, and then subtracting the constant function $\mathbb{I}$ we have $\overline{\mathbf{W}}_0(\diamond;x,t)-\mathbb{I}\in C^\nu(\Gamma\setminus D)$.

Restricting now to $\Gamma\cap D$, we write $\overline{\mathbf{W}}_0(s;x,t)-\mathbb{I}=\ee^{-\ii(\lambda(s)x+\lambda(s)^2t)\sigma_3}\lambda(s)^{-4-\varepsilon}
\cdot \lambda(s)^{4+\varepsilon}(\mathbf{V}_0(\lambda(s))-\mathbb{I})\ee^{\ii(\lambda(s)x+\lambda(s)^2t)\sigma_3}$, where the branch cuts of the functions $\lambda(s)^{\pm(4+\varepsilon)}$ emanate from the simple pole $s=1$ to the left (hence only touching $\Gamma\cap D$ at $s=1$).  According to \eqref{Holder-infinity}, it follows that $\mathbf{B}(\diamond):=\lambda(\diamond)^{4+\varepsilon}(\mathbf{V}_0(\lambda(\diamond))-\mathbb{I})\in C^\nu(\Gamma\cap D)$ for $\delta>0$ small enough.  The diagonal part of $\overline{\mathbf{W}}_0(\diamond;x,t)-\mathbb{I}$ is equal to the product of the scalar function $\lambda(\diamond)^{-4-\varepsilon}$ (obviously H\"older-continuous with exponent $\nu>0$ because $\epsilon>0$ and $s=1$ is a simple pole of $s\mapsto\lambda(s)$) and the diagonal part of $\mathbf{B}(\diamond)$; hence H\"older continuity of the diagonal part is confirmed on $\Gamma\cap D$.  For the off-diagonal part, it then suffices to show that the scalar functions $f_\pm(s):=\lambda(s)^{-4-\varepsilon}\ee^{\pm 2\ii(\lambda(s)x+\lambda(s)^2t)}$ are H\"older continuous with exponent $\nu$ on $\Gamma\cap D$.  We compute 
\begin{equation}
f'_\pm(s)=\left[-(4+\varepsilon)\lambda(s)^{-1}\pm 2\ii(x+2\lambda(s)t)\right]\lambda(s)^{-4-\varepsilon}\lambda'(s)\ee^{\pm 2\ii(\lambda(s)x+\lambda(s)^2t)}.
\label{eq:fprime}
\end{equation}
Since $\lambda(s)\in\mathbb{R}$ for $s\in\Gamma\cap D$ when the radius $\delta>0$ is sufficiently small, and since $s\mapsto\lambda(s)$ has a simple pole at $s=1$, we easily find that $|f'_\pm(s)|=O(|s-1|^{1+\varepsilon})$ on $\Gamma\cap D$, which is bounded near $s=1$.  It follows that the functions $f_\pm(\diamond)$ are both Lipschitz and hence H\"older continuous on $\Gamma\cap D$ with exponent $\nu$.   This shows that $\overline{\mathbf{W}}_0(\diamond;x,t)-\mathbb{I}\in C^\nu(\Gamma\cap D)$.  

Since we have now shown that $\overline{\mathbf{W}}_0(\diamond;x,t)-\mathbb{I}\in C^\nu(\Gamma\cap D)\cap C^\nu(\Gamma\setminus D)$ and since the matrix function in question is continuous at the junction points  $s\in\Gamma$ with $|s-1|=\delta$, the proof is complete.
\end{proof}

This result implies that the Beals-Coifman operator $\mathcal{T}(x,t)$ is a bounded linear operator on $C^\mu(\Gamma)$ whenever $0<\mu\leq \nu<1$. We denote by $\mathcal{I}$ the identity operator on $C^\mu(\Gamma)$. For arbitrary given $(x,t)\in\mathbb{R}^2$, every solution $\mathbf{X}\in C^\mu(\Gamma)$ of the Beals-Coifman singular integral equation
\begin{equation}
(\mathcal{I} - \mathcal{T}(x,t))\mathbf{X} = \mathbb{I} \in C^\mu(\Gamma)
\end{equation}
produces a solution of \rhref{rhp:compact} via the formula
\begin{equation}
\begin{split}
\mathbf{N}(z;x,t) &= \mathbb{I}+\mathcal{C}^\Gamma[\mathbf{X}(\diamond;x,t)(\overline{\mathbf{W}}_0(\diamond;x,t)-\mathbb{I})](z)\\
&=\mathbb{I} + \frac{1}{2 \pi \ii} \int_{\Gamma}\frac{\mathbf{X}(s;x,t)(\overline{\mathbf{W}}_0(s;x,t) - \mathbb{I})}{s-z}\dd s,\quad z\in\mathbb{C}\setminus\Gamma.
\end{split}
\label{N-Cauchy}
\end{equation}
We know from Fredholm theory (see \cite[Section A.5]{KamvissisMM2003}) combined with the equivalence of \rhref{rhp:background} and \rhref{rhp:compact} and the vanishing lemma that if the strict inequalities $0<\mu<\nu<1$ hold, $\mathcal{I}-\mathcal{T}(x,t)$ is invertible with a bounded inverse on $C^\mu(\Gamma)$ for any given $(x,t)\in\mathbb{R}^2$. We set $\mathbf{X}(s;x,t):= (\mathcal{I} - \mathcal{T}(x,t))^{-1}[\mathbb{I}](s)$.

A result related to Lemma~\ref{lem:HoelderJumpGamma} that we will need later is the following.
\begin{lemma}
As $(x,t)\to(x_0,t_0)\in\mathbb{R}^2$, $\overline{\mathbf{W}}_0(\diamond;x,t)\to\overline{\mathbf{W}}_0(\diamond;x_0,t_0)$ in $C^\nu(\Gamma)$.
\label{lem:convergence}
\end{lemma}
\begin{proof}
 As in the proof of Lemma~\ref{lem:HoelderJumpGamma}, we again write $\Gamma$ as the disjoint union $\Gamma=(\Gamma\setminus D)\sqcup(\Gamma\cap D)$.   The $C^\nu(\Gamma)$ norm of $\overline{\mathbf{W}}_0(\diamond;x,t)-\overline{\mathbf{W}}_0(\diamond;x_0,t_0)$ is controlled by the sum of its $C^\nu(\Gamma\setminus D)$ and $C^\nu(\Gamma\cap D)$ norms.  Since derivatives of $\overline{\mathbf{W}}_0(\diamond;x,t)$ with respect to $(x,t)$ are bounded on $\Gamma\setminus D$ where $s\mapsto\lambda(s)$ is analytic (uniformly for $(x,t)$ in a neighborhood of $(x_0,t_0)$), it is obvious that $\overline{\mathbf{W}}_0(\diamond;x,t)\to\overline{\mathbf{W}}_0(\diamond;x_0,t_0)$ in $C^\nu(\Gamma\setminus D)$ as $(x,t)\to (x_0,t_0)$ in $\mathbb{R}^2$.  

Now we consider $\Gamma\cap D$.
The diagonal elements of $\overline{\mathbf{W}}_0(\diamond;x,t)$ are independent of $(x,t)\in\mathbb{R}^2$, so it suffices to consider the off-diagonal elements which, as shown in the proof of Lemma~\ref{lem:HoelderJumpGamma}, are proportional on $\Gamma\cap D$ to the functions $f_\pm(\diamond)=f_\pm(\diamond;x,t)$ via factors that are independent of $(x,t)\in\mathbb{R}^2$ and H\"older continuous on $\Gamma\cap D$ with exponent $\nu$.  We set $\Delta f_\pm(\diamond):=f_\pm(\diamond;x,t)-f_\pm(\diamond;x_0,t_0)$ and then for $s_1,s_2\in\Gamma$ we have
\begin{equation}
\Delta f_\pm(s_2)-\Delta f_\pm(s_1)=\int_{s_1}^{s_2}\Delta f'(s)\,\dd s = \int_{s_1}^{s_2}(f_\pm'(s;x,t)-f_\pm'(s;x_0,t_0))\,\dd s.
\end{equation}
Now, taking $L$ to be a path in $\mathbb{R}^2$ from $(x_0,t_0)$ to $(x,t)$, we may write
\begin{equation}
f_\pm'(s;x,t)-f_\pm'(s;x_0,t_0)=\int_L\nabla f_\pm'(s;\bar{x},\bar{t})\cdot\dd\ell(\bar{x},\bar{t})
\end{equation}
where $\nabla$ is the gradient with respect to the parameters $(\bar{x},\bar{t})\in L$, and $\dd\ell(\bar{x},\bar{t})$ denotes the oriented differential line element.  From \eqref{eq:fprime} and the fact that $s\mapsto\lambda(s)$ has a simple pole at $s=1$, it is easy to estimate that the Euclidean length of $\nabla f_\pm'(s;\bar{x},\bar{t})$ is uniformly $O(|s-1|^{\varepsilon-1})$ for $s\in\Gamma\cap D$ and $(x,t)$ near $(x_0,t_0)$.  Therefore, $f_\pm'(s;x,t)-f_\pm'(s;x_0,t_0)=O(|L||s-1|^{\varepsilon-1})$ where $|L|$ denotes the arc length of $L$.  Since this estimate is integrable at $s=1\in \Gamma\cap D$,  we see that $\Delta f_\pm(s_2)-\Delta f_\pm(s_1) = O(|L||s_2-s_1|^\epsilon)$.  Using $\epsilon\ge \nu$, it then follows that the $C^\nu(\Gamma\cap D)$ norm of $\Delta f_\pm(\diamond)$ is proportional to $|L|$ and hence vanishes as $(x,t)\to (x_0,t_0)$ in $\mathbb{R}^2$.  Therefore $\overline{\mathbf{W}}_0(\diamond;x,t)\to\overline{\mathbf{W}}_0(\diamond;x_0,t_0)$ in $C^\nu(\Gamma\cap D)$ in the same limit, and the proof is finished.
\end{proof}

\begin{proof}[Proof of Proposition~\ref{prop:independence}]
Since $\mathbf{X}\in C^\mu(\Gamma)$ and $0<\mu<\nu<1$, Lemma~\ref{lem:HoelderJumpGamma} implies that for $(x,t)=(x_0,t_0)\in\mathbb{R}^2$, the density in the Cauchy integral \eqref{N-Cauchy} lies in $C^\mu(\Gamma)$.  By the Plemelj-Privalov theorem it follows that $\mathbf{N}(\diamond;x_0,t_0)$ is H\"older continuous with exponent $\mu$ uniformly on each connected component of $\mathbb{C}\setminus\Gamma$.  In particular, since $\overline{\mathbf{W}}_0(1;x_0,t_0)=\mathbb{I}$,  $z\mapsto \mathbf{N}(z;x_0,t_0)$ is continuous at $z=1$ even though generally it has a jump across the unit circle $\mathbb{T}\subset\Gamma$.  Thus $\mathbf{N}(1;x_0,t_0)$ is well-defined and has unit determinant, and from $\mathbf{M}(\lambda;x_0,t_0):=\mathbf{N}(1;x_0,t_0)^{-1}\mathbf{N}(z(\lambda);x_0,t_0)$ and $\lambda(z)\to\infty$ as $z\to 1$ the result follows.
\end{proof}

\begin{proof}[Proof of Proposition~\ref{prop:continuity}] 
For brevity, we denote by $\mathcal{A}:=\mathcal{I}-\mathcal{T}(x_0,t_0)$ the Beals-Coifman operator anchored at the point $(x_0,t_0)$, and we denote by $\mathcal{B}:=\mathcal{I}-\mathcal{T}(x,t)$ the Beals-Coifman operator that varies with $(x,t)$. Thus, with $0<\mu<\nu<1$, the bounded linear operators $\mathcal{A}$ and $\mathcal{B}$ are invertible on $C^\mu(\Gamma)$ with bounded inverses. To control the operator norm $\| \mathcal{B}^{-1} \|_{C^\mu(\Gamma)\circlearrowleft}$, we follow \cite[Chapter 17]{Lax2002} and write
\begin{equation}
\mathcal{B} = \mathcal{A}\left[ \mathcal{I} - \mathcal{A}^{-1}(\mathcal{A}-\mathcal{B})\right].
\end{equation}
We make the proximity assumption
\begin{equation}
\| \mathcal{B} - \mathcal{A}\|_{C^\mu(\Gamma)\circlearrowleft} < \frac{1}{\| \mathcal{A}^{-1}\|_{C^\mu(\Gamma)\circlearrowleft}}
\label{B-A}
\end{equation}
on $\mathcal{B}$,
which implies that $\| \mathcal{A}^{-1} (\mathcal{A}-\mathcal{B}) \|_{C^\mu(\Gamma)\circlearrowleft}<1$. Thus, we may express $\mathcal{B}^{-1}$ via a Neumann series
\begin{equation}
\mathcal{B}^{-1} =\left[ \mathcal{I} - \mathcal{A}^{-1}(\mathcal{A}-\mathcal{B})\right]^{-1} \mathcal{A}^{-1} = \left[\sum_{k=0}^\infty \left(\mathcal{A}^{-1} (\mathcal{A}-\mathcal{B}) \right)^k \right] \mathcal{A}^{-1},
\end{equation}
convergent in operator norm on $C^\mu(\Gamma)$, and obtain the bound
\begin{equation}
\| \mathcal{B}^{-1} \|_{C^\mu(\Gamma)\circlearrowleft} \leq \frac{\|  \mathcal{A}^{-1} \|_{C^\mu(\Gamma)\circlearrowleft}}{1- \| \mathcal{A}^{-1}(\mathcal{A}-\mathcal{B})\|_{C^\mu(\Gamma)\circlearrowleft}}  \leq \frac{ \|  \mathcal{A}^{-1} \|_{C^\mu(\Gamma)\circlearrowleft}}{1- \| \mathcal{A}^{-1}\|_{C^\mu(\Gamma)\circlearrowleft} \| \mathcal{B}-\mathcal{A}\|_{C^\mu(\Gamma)\circlearrowleft}}
\label{B-inv-bound}
\end{equation}
on the operator norm of $\mathcal{B}^{-1}$.
Then we have
\begin{equation}
\| \mathcal{B}^{-1} -\mathcal{A}^{-1} \|_{C^\mu(\Gamma)\circlearrowleft} =\| \mathcal{B}^{-1}(\mathcal{A} - \mathcal{B})\mathcal{A}^{-1} \|_{C^\mu(\Gamma)\circlearrowleft} \leq \| \mathcal{B}^{-1} \|_{C^\mu(\Gamma)\circlearrowleft} \| \mathcal{B} - \mathcal{A} \|_{C^\mu(\Gamma)\circlearrowleft} \| \mathcal{A}^{-1} \|_{C^\mu(\Gamma)\circlearrowleft}
\end{equation}
and from \eqref{B-inv-bound} we obtain
\begin{equation}
\| \mathcal{B}^{-1} -\mathcal{A}^{-1} \|_{C^\mu(\Gamma)\circlearrowleft} \leq  \frac{ \|  \mathcal{A}^{-1} \|^2_{C^\mu(\Gamma)}  \| \mathcal{B}-\mathcal{A}\|_{C^\mu(\Gamma)\circlearrowleft}}{1- \| \mathcal{A}^{-1}\|_{C^\mu(\Gamma)\circlearrowleft} \| \mathcal{B}-\mathcal{A}\|_{C^\mu(\Gamma)\circlearrowleft}}.
\label{Bi-Ai}
\end{equation}
It is clear from \eqref{Bi-Ai} that $\| \mathcal{B}^{-1} -\mathcal{A}^{-1} \|_{C^\mu(\Gamma)\circlearrowleft}$ can be made arbitrarily small by making $\| \mathcal{B}-\mathcal{A}\|_{C^\mu(\Gamma)\circlearrowleft}$ small.   Now explicitly,
\begin{equation}
\begin{split}
(\mathcal{B}-\mathcal{A})[\dummy](s)&=\mathcal{C}_-^\Gamma[\dummy(\diamond)(\overline{\mathbf{W}}_0(\diamond;x,t)-\overline{\mathbf{W}}_0(\diamond;x_0,t_0))](s)\\
&=\frac{1}{2\pi\ii}\int_\Gamma\frac{\dummy(w)(\overline{\mathbf{W}}_0(w;x,t)-\overline{\mathbf{W}}_0(w;x_0,t_0))}{w-s_-}\,\dd w,
\end{split}
\end{equation}
so
\begin{equation}
\|\mathcal{B}-\mathcal{A}\|_{C^\mu(\Gamma)\circlearrowleft}\le \|\mathcal{C}^\Gamma_-\|_{C^\mu(\Gamma)\circlearrowleft}\|\overline{\mathbf{W}}_0(\diamond;x,t)-\overline{\mathbf{W}}_0(\diamond;x_0,t_0)\|_{C^\mu(\Gamma)}\lesssim \|\overline{\mathbf{W}}_0(\diamond;x,t)-\overline{\mathbf{W}}_0(\diamond;x_0,t_0)\|_{C^\nu(\Gamma)},
\end{equation}
because the Cauchy projector $\mathcal{C}^\Gamma_-$ is bounded on $C^\mu(\Gamma)$ by the Plemelj-Privalov theorem and $0<\mu<\nu<1$.  Since the implied constants are independent of $(x,t)$, it follows from Lemma~\ref{lem:convergence} that $\|\mathcal{B}-\mathcal{A}\|_{C^\mu(\Gamma)\circlearrowleft}\to 0$ as $(x,t)\to (x_0,t_0)$ in $\mathbb{R}^2$.  According to \eqref{Bi-Ai}, it then follows that also $\|\mathcal{B}^{-1}-\mathcal{A}^{-1}\|_{C^\mu(\Gamma)\circlearrowleft}\to 0$ as $(x,t)\to (x_0,t_0)$ in $\mathbb{R}^2$ because $\mathcal{A}^{-1}$ is independent of $(x,t)$.

According to \eqref{N-Cauchy}, 
\begin{equation}
\begin{split}
\mathbf{N}(z;x,t)-\mathbf{N}(z;x_0,t_0)&=\mathcal{C}^\Gamma[\mathbf{X}(\diamond;x,t)(\overline{\mathbf{W}}_0(\diamond;x,t)-\mathbb{I})](z)-\mathcal{C}^\Gamma[\mathbf{X}(\diamond;x_0,t_0)(\overline{\mathbf{W}}_0(\diamond;x_0,t_0)-\mathbb{I})](z)\\
&=\mathcal{C}^\Gamma[\mathbf{X}(\diamond;x,t)(\overline{\mathbf{W}}_0(\diamond;x,t)-\mathbb{I})-\mathbf{X}(\diamond;x_0,t_0)(\overline{\mathbf{W}}_0(\diamond;x_0,t_0)-\mathbb{I})](z),
\end{split}
\end{equation}
where $\mathbf{X}(s;x,t)=\mathcal{B}^{-1}[\mathbb{I}](s)$ and $\mathbf{X}(s;x_0,t_0)=\mathcal{A}^{-1}[\mathbb{I}](s)$ for $s\in\Gamma$.  Again by the Plemelj-Privalov theorem, the H\"older norm of $\mathcal{C}^\Gamma[\mathbf{X}](z)$ with exponent $\mu$ taken for $z$ ranging over any connected component of $\mathbb{C}\setminus\Gamma$ is bounded by a multiple of $\|\mathbf{X}\|_{C^\mu(\Gamma)}$.  Therefore, since the sup norm is dominated by the H\"older norm,
\begin{equation}
\sup_{z\in\mathbb{C}\setminus\Gamma}\|\mathbf{N}(z;x,t)-\mathbf{N}(z;x_0,t_0)\|\lesssim \|\mathbf{X}(\diamond;x,t)(\overline{\mathbf{W}}_0(\diamond;x,t)-\mathbb{I})-\mathbf{X}(\diamond;x_0,t_0)(\overline{\mathbf{W}}_0(\diamond;x_0,t_0)-\mathbb{I})\|_{C^\mu(\Gamma)}.
\end{equation}
By adding and subtracting terms appropriately, 
\begin{multline}
\mathbf{X}(\diamond;x,t)\left(\overline{\mathbf{W}}_0(\diamond;x,t) - \mathbb{I}\right) - \mathbf{X}(\diamond;x_0,t_0)\left(\overline{\mathbf{W}}_0(\diamond;x_0,t_0) - \mathbb{I}\right) =\\
 (\mathbf{X}(\diamond;x_0,t_0) - \mathbf{X}(\diamond;x,t))+  \left( \mathbf{X}(\diamond;x,t) -   \mathbf{X}(\diamond;x_0,t_0)\right) \left(\overline{\mathbf{W}}_0(\diamond;x,t) - \overline{\mathbf{W}}_0(\diamond;x_0,t_0) \right) \\+ \left( \mathbf{X}(\diamond;x,t) -   \mathbf{X}(\diamond;x_0,t_0)\right) \overline{\mathbf{W}}_0(\diamond;x_0,t_0) 
+ \mathbf{X}(\diamond;x_0,t_0) \left(\overline{\mathbf{W}}_0(\diamond;x,t) -  \overline{\mathbf{W}}_0(\diamond;x_0,t_0)\right).
\end{multline}
But $\|\mathcal{B}^{-1}-\mathcal{A}^{-1}\|_{C^\mu(\Gamma)\circlearrowleft}\to 0$ as $(x,t)\to (x_0,t_0)$ in $\mathbb{R}^2$ implies that $\|\mathbf{X}(\diamond;x,t)-\mathbf{X}(\diamond;x_0,t_0)\|_{C^\mu(\Gamma)}\to 0$ in the same limit, and since $0<\mu<\nu<1$, Lemma~\ref{lem:convergence} implies $\|\overline{\mathbf{W}}_0(\diamond;x,t)-\overline{\mathbf{W}}_0(\diamond;x_0,t_0)\|_{C^\mu(\Gamma)}\to 0$ in the same limit as well.  Since Lemma~\ref{lem:HoelderJumpGamma} implies that $\overline{\mathbf{W}}_0(\diamond;x_0,t_0)\in C^\nu(\Gamma)\subset C^\mu(\Gamma)$, and since $C^\mu(\Gamma)$ is a Banach algebra, it then follows that
\begin{equation}
\lim_{(x,t)\to (x_0,t_0)}\sup_{z\in\mathbb{C}\setminus\Gamma}\|\mathbf{N}(z;x,t)-\mathbf{N}(z;x_0,t_0)\|=0.
\label{eq:N-convergence}
\end{equation}
By the identity $\mathbf{M}(\lambda;x,t)=\mathbf{N}(1;x,t)^{-1}\mathbf{N}(z(\lambda);x,t)$ valid for all $(x,t)\in\mathbb{R}^2$, it then follows also that
\begin{equation}
\begin{split}
\mathbf{M}(\lambda;x,t)-\mathbf{M}(\lambda;x_0,t_0)&=\mathbf{N}(1;x,t)^{-1}\mathbf{N}(z(\lambda);x,t)-\mathbf{N}(1;x_0,t_0)^{-1}\mathbf{N}(z(\lambda);x_0,t_0)\\
&=
(\mathbf{N}(1;x,t)^{-1}-\mathbf{N}(1;x_0,t_0)^{-1})(\mathbf{N}(z(\lambda);x,t)-\mathbf{N}(z(\lambda);x_0,t_0))\\
&\quad{}+(\mathbf{N}(1;x,t)^{-1}-\mathbf{N}(1;x_0,t_0)^{-1})\mathbf{N}(z(\lambda);x_0,t_0)\\
&\quad{}+\mathbf{N}(1;x_0,t_0)^{-1}(\mathbf{N}(z(\lambda);x,t)-\mathbf{N}(z(\lambda);x_0,t_0)),
\end{split}
\end{equation}
so, using the facts that $\det(\mathbf{N}(z;x,t))=1$ for all $z\in\mathbb{C}\setminus\Gamma$ and $(x,t)\in\mathbb{R}^2$, and that $\mathbf{N}(1;x,t)$ is well-defined, the use of \eqref{eq:N-convergence} completes the proof.
\end{proof}

\section{Extreme Superposition}
\label{sec:superposition}
In this section we carry out a nonlinear superposition procedure and place coherent structures on top of the background field $q_0(x,t)$ given by \rhref{rhp:background}. 
Fix $(x_0,t_0)\in\mathbb{R}^2$, recall the radius $r>\sup |\Sigma\setminus \mathbb{R}| $ from \eqref{Mobius} and let $\Sigma_\circ$ denote the circle $|\lambda|=r$. Define the arcs
\begin{equation}
\begin{aligned}
  \Sigma_\circ^{\pm} &:= \lbrace \lambda\in \Sigma_\circ \colon \pm \Im\{\lambda\}\geq 0 \rbrace,\\
  \Sigma_{\mathrm{L}} &:= \lbrace \lambda\in \mathbb{R} \colon \lambda\in(-\infty,-r] \rbrace,\\
  \Sigma_{\mathrm{R}} &:= \lbrace \lambda\in \mathbb{R} \colon \lambda\in[r,+\infty) \rbrace.
\end{aligned}
\end{equation}
We orient both semicircular arcs $\Sigma_\circ^{\pm}$ in the direction from $\lambda=- r$ to $\lambda= r$, $\Sigma_{\mathrm{L}}$ from $\lambda=-\infty$ to $\lambda=- r$, and $\Sigma_{\mathrm{R}}$ from $\lambda=r$ to $\lambda=+\infty$. 
To describe the regions separated by this system of contours, set
\begin{equation}
  D_{\infty}^{\pm}:= \lbrace \lambda\in\mathbb{C}\colon |\lambda|>r~\text{and}~\pm\Im\{\lambda\} > 0 \rbrace,
\end{equation}
and let $D_\circ$ denote the open disk whose boundary is $\Sigma_\circ$.

Define from the solution of \rhref{rhp:background} the related matrix function
\begin{equation}
\mathbf{U}(\lambda;x,t) : = \mathbf{M}(\lambda;x,t) \ee^{\ii \lambda(x + \lambda t)\sigma_3},
\end{equation}
which is a fundamental matrix of simultaneous solutions of the system of Lax pair equations \eqref{Lax-pair}.
As in \cite[Proposition 2.1]{BilmanM2019}, we also define the simultaneous solution matrix $\mathbf{F}(\lambda;x,t)$ for the system of Lax pair equations \eqref{Lax-pair} that is normalized to satisfy
\begin{equation}
\mathbf{F}(\lambda;x_0,t_0) = \mathbb{I}.
\label{F-normalization}
\end{equation}
Note that $\mathbf{F}(\lambda;x,t)$ can be constructed in practice by  
\begin{equation}
\mathbf{F}(\lambda;x,t) = \mathbf{U}(\lambda;x,t) \mathbf{U}(\lambda;x_0,t_0)^{-1}.
\label{F-to-U}
\end{equation} 
Moreover, $\mathbf{F}(\lambda;x,t)$ is unimodular and entire in $\lambda$ for $\lambda\in\mathbb{C}$, see \cite[Proposition 2.1]{BilmanM2019}. 
We now introduce the following piecewise-defined solution matrix for the system of Lax pair equations \eqref{Lax-pair}:
\begin{equation}
\mathbf{U}^{[0]}(\lambda;x,t) :=
\begin{cases} 
\mathbf{U}(\lambda;x,t),&\quad \lambda \in D_\infty^+ \cup D_\infty^-,\\
\mathbf{F}(\lambda;x,t),&\quad \lambda \in D_\circ.
\end{cases}
\label{U-0}
\end{equation}
Recalling that $\Sigma_\circ^+$ and $\Sigma_\circ^+$ are both oriented from $-r$ to $r$, the relation \eqref{F-to-U} implies the following jump conditions satisfied by the (continuous) boundary values of $\mathbf{U}^{[0]}(\lambda;x,t)$ as $\lambda$ tends to a point $\zeta\in\Sigma_\circ$:
\begin{align}
\mathbf{U}^{[0]}_{+}(\zeta;x,t) = \mathbf{U} (\zeta;x,t) = \mathbf{F}(\zeta;x,t) \mathbf{U} (\zeta;x_0,t_0) = \mathbf{U}^{[0]}_{-}(\zeta;x,t) \mathbf{U} (\zeta;x_0,t_0),\qquad \zeta \in \Sigma_\circ^+,
\end{align}
and similarly,
\begin{align}
\mathbf{U}^{[0]}_{+}(\zeta;x,t) = \mathbf{F} (\zeta;x,t) = \mathbf{U}(\zeta;x,t) \mathbf{U} (\zeta;x_0,t_0)^{-1} = \mathbf{U}^{[0]}_{-}(\zeta;x,t) \mathbf{U} (\zeta;x_0,t_0)^{-1},\qquad \zeta \in \Sigma_\circ^-.
\end{align}
The related matrix
\begin{equation}
\mathbf{M}^{[0]}(\lambda;x,t) := \mathbf{U}^{[0]}(\lambda;x,t)  \ee^{\ii \lambda(x+\lambda t)\sigma_3}
\label{M-0}
\end{equation}
is analytic in $\lambda$ for $\lambda \in \mathbb{C}\setminus \Sigma^\sharp$, where $\Sigma^\sharp:=\Sigma_\circ \cup \Sigma_{\mathrm{L}} \cup \Sigma_{\mathrm{R}}$ with 
H\"older continuous boundary values on $\Sigma^\sharp\setminus\{-r,r\}$ for all exponents $\mu<\nu$ satisfying
\begin{equation}
\mathbf{M}^{[0]}_+(\zeta;x,t)=\mathbf{M}^{[0]}_-(\zeta;x,t) \ee^{-\ii \zeta (x + \zeta t)\sigma_3} \mathbf{V}^{[0]}(\zeta) \ee^{\ii \zeta (x + \zeta t)\sigma_3},\qquad \zeta \in \Sigma^\sharp,
\label{M-0-jump}
\end{equation}
where
\begin{equation}
\mathbf{V}^{[0]}(\zeta):=
\begin{cases}
\mathbf{U}(\zeta;x_0,t_0),&\quad \zeta \in \Sigma_\circ^+,\\
\mathbf{U}(\zeta;x_0,t_0)^{-1},&\quad \zeta \in \Sigma_\circ^-,\\
\mathbf{V}^{\mathbb{R}}(\zeta),&\quad \zeta \in \Sigma_{\mathrm{L}}\cup \Sigma_{\mathrm{R}}.
\end{cases}
\label{V-0}
\end{equation}
Moreover, $\mathbf{M}^{[0]}(\lambda;x,t) \to \mathbb{I}$ as $\lambda\to\infty$ in $\mathbb{C}\setminus \Sigma^\sharp$ and $\det(\mathbf{M}^{[0]}(\lambda;x,t))=1$. Since $\mathbf{M}^{[0]}(\lambda;x,t) = \mathbf{M}(\lambda;x,t)$ for $|\lambda|>r$, we have
\begin{equation}
 q_0(x,t)= 2\ii  \lim_{\lambda\to\infty} \lambda M^{[0]}_{12}(\lambda;x,t).
\label{q-0-recovery}
\end{equation}
{
\begin{remark}
Another type of background field that can be included in this framework is a solution $q=q_0(x,t)$ of the Cauchy problem for \eqref{nls} with nonzero boundary conditions $q_0(x,t)-\ee^{\ii t}\to 0$ as $x\to\pm\infty$.  To handle this kind of background solution, instead of starting with \rhref{rhp:background} and converting it to $\mathbf{M}^{[0]}(\lambda;x,t)$ as indicated above, we may begin with the matrix $\mathbf{M}(\lambda;x,t)$ satisfying the conditions of \cite[Riemann-Hilbert Problem 1]{BilmanM2019} which encodes the initial data $q_0(x,0)$ into $q_0(x,t)$.  Then we simply define $\mathbf{M}^{[0]}(\lambda;x,t):=\ee^{\frac{1}{2}\ii t\sigma_3}\mathbf{M}(\lambda;x,t)\ee^{\ii (\lambda-\rho(\lambda))(x+\lambda t)\sigma_3}$.  Here $\rho(\lambda)$ is the function analytic except on the line segment connecting the points $\lambda=\pm\ii$ that is defined by the conditions $\rho(\lambda)^2=\lambda^2+1$ and $\rho(\lambda)=\lambda+O(\lambda^{-1})$ as $\lambda\to\infty$.  The latter function appears naturally in the context of the Beals-Coifman solutions appropriate for nonzero boundary conditions, but the transformation to $\mathbf{M}^{[0]}(\lambda;x,t)$ effectively replaces this function by $\lambda$.
\label{rem:nonzero-background}
\end{remark}
}
\subsection{Nonlinear superposition scheme}
Fix an arbitrary point $\xi$ in the upper $\lambda$-half-plane such that $|\xi|<r$, i.e., $\xi \in D_\circ$. Note that $r$ can be chosen arbitrarily large since all it needs to satisfy is that $r>\sup |\Sigma \setminus \mathbb{R}|$, hence $\xi\in\mathbb{C}^+$ is effectively arbitrary. Now let $\mathbf{G}^{[0]}(\lambda;x,t)$ denote the Darboux transformation matrix $\mathbf{G}_{\infty}(\lambda;x,t)$ constructed in \cite[Section 3.2, Eqns.\@ (3.37)--(3.39)]{BilmanM2019} for the focusing NLS equation in the context of the \emph{robust} modification of the inverse-scattering transform introduced therein. 
$\mathbf{G}^{[0]}(\lambda;x,t)$ is constructed from the \emph{seed} column-vector simultaneous solution of the Lax pair equations \eqref{Lax-pair} with the potential $q= q_0(x,t)$ given by
\begin{equation}
\mathbf{s}(x,t):= \mathbf{U}^{[0]}(\xi; x, t)\mathbf{c},
\end{equation}
where $\mathbf{c}\in \mathbb{C}^2\setminus \{ \mathbf{0}\}$ is a given parameter vector,
and $\mathbf{G}^{[0]}(\lambda;x,t)$ is 
of the form
\begin{equation}
\mathbf{G}^{[0]}(\lambda;x,t) = \mathbb{I} + \frac{\mathbf{Y}^{[0]}(x,t)}{\lambda - \xi} + \frac{\mathbf{Z}^{[0]}(x,t)}{\lambda - \xi^*}.
\label{G-0}
\end{equation}
Here $\mathbf{Y}^{[0]}(x,t)$ and $\mathbf{Z}^{[0]}(x,t)$ are $2\times 2$ matrices which satisfy $\mathbf{Z}^{[0]}(x,t) = \sigma_2\mathbf{Y}^{[0]}(x,t)^*\sigma_2 $, and
\begin{multline}
\mathbf{Y}^{[0]}(x,t):= - \frac{4\beta^2 w(x,t)^*}{4\beta^2 |w(x,t)|^2 + N(x,t)^2} \mathbf{s}(x,t)\mathbf{s}(x,t)^\top \sigma_2 \\+ \frac{2 \ii \beta N(x,t)}{4\beta^2 |w(x,t)|^2 + N(x,t)^2} \sigma_2  \mathbf{s}(x,t)^* \mathbf{s}(x,t)^\top \sigma_2,
\label{Y}
\end{multline}
where
\begin{equation}
\beta:= \Im(\xi) >0, \qquad N(x,t):= \mathbf{s}(x,t)^\dagger  \mathbf{s}(x,t) = \|  \mathbf{s}(x,t) \|^2 >0,\qquad w(x,t) := \mathbf{s}(x,t)^\top \sigma_2  \mathbf{s}'(x,t),
\label{N-w}
\end{equation}
and $ \mathbf{s}'(x,t) := \frac{\dd}{\dd \lambda}\mathbf{U}^{[0]}(\lambda;x,t)\mathbf{c}\rvert_{\lambda=\xi}$. Note that the strict positivity of $N(x,t)$ follows from the fact that $\mathbf{U}^{[0]}(\xi;x,t)$ is unimodular. 
{
We see from \eqref{Y} that $\mathbf{G}^{[0]}(\lambda;x,t)$ does not depend on the Euclidean length $\|\mathbf{c}\|$ of the complex vector $\mathbf{c}=:[c_1\quad c_2]^\top$. Therefore, we may consider $\mathbf{c}$ to be an element of the complex projective space $\mathbb{CP}^1$ and write $\mathbf{c}=[c_1 \colon c_2]$.
}

Now, define 
\begin{equation}
\mathbf{U}^{[1]}(\lambda;x,t):=
\begin{cases}
\mathbf{G}^{[0]}(\lambda;x,t) \mathbf{U}^{[0]}(\lambda;x,t),&\quad \lambda \in D_{\infty}^+\cup D_{\infty}^-,\\
\mathbf{G}^{[0]}(\lambda;x,t) \mathbf{U}^{[0]}(\lambda;x,t) \mathbf{G}^{[0]}(\lambda;x_0,t_0)^{-1},&\quad \lambda \in D_{\circ}.
\end{cases}
\end{equation}
As was shown in \cite[Section 3.2]{BilmanM2019}, $\mathbf{U}^{[1]}(\lambda;x,t)$ satisfies the simultaneous system of Lax pair equations \eqref{Lax-pair} in which $q=q_0(x,t)$ is replaced with the potential $q=q^{[1]}(x,t)$ given by
\begin{equation}
q_1(x,t):= q_0(x,t) + 2 \ii \left(Y_{12}(x, t)-Y_{21}(x, t)^{*}\right).
\label{q1-def}
\end{equation}
Recall the definition \eqref{U-0} and observe that for $\lambda\in D_\circ$ we have
\begin{equation}
\mathbf{U}^{[1]}(\lambda;x,t) = \mathbf{G}^{[0]}(\lambda;x,t) \mathbf{F}(\lambda;x,t) \mathbf{G}^{[0]}(\lambda;x_0,t_0)^{-1},
\end{equation}
which extends to an analytic function of $\lambda$ at $\lambda=\xi$ and $\lambda=\xi^*$ and also maintains the property $\mathbf{U}^{[1]}(\lambda;x_0,t_0)=\mathbb{I}$ for $\lambda\in D_\circ$.
Thus, the matrix function
\begin{equation}
\mathbf{M}^{[1]}(\lambda;x,t) := \mathbf{U}^{[1]}(\lambda;x,t) \ee^{\ii \lambda(x+\lambda t)\sigma_3}
\end{equation}
is analytic in $\lambda$ for $\lambda\in \mathbb{C}\setminus \Sigma^\sharp$, that is, for $\lambda \in  D_\circ \cup D_{\infty}^+\cup  D_{\infty}^-$. It is easy to see from the jump conditions \eqref{M-0-jump}--\eqref{V-0} satisfied by $\mathbf{M}^{[0]}(\lambda;x,t)$ that $\mathbf{M}^{[1]}(\lambda;x,t)$ admits continuous boundary values on $\Sigma^\sharp$ that are related by the modified jump condition
\begin{equation}
\mathbf{M}^{[1]}_+(\zeta;x,t)=\mathbf{M}^{[1]}_-(\zeta;x,t) \ee^{-\ii \zeta (x + \zeta t)\sigma_3} \mathbf{V}^{[1]}(\zeta) \ee^{\ii \zeta (x + \zeta t)\sigma_3},\qquad \zeta \in \Sigma^\sharp,
\end{equation}
where
\begin{equation}
\mathbf{V}^{[1]}(\zeta):=
\begin{cases}
\mathbf{G}^{[0]}(\lambda;x_0,t_0) \mathbf{U}(\zeta;x_0,t_0),&\quad \zeta \in \Sigma_\circ^+,\\
\mathbf{U}(\zeta;x_0,t_0)^{-1}\mathbf{G}^{[0]}(\lambda;x_0,t_0)^{-1},&\quad \zeta \in \Sigma_\circ^-,\\
\mathbf{V}^{\mathbb{R}}(\zeta),&\quad \zeta \in \Sigma_{\mathrm{L}}\cup \Sigma_{\mathrm{R}},
\end{cases}
\label{V-1}
\end{equation}
and the transformation $\mathbf{M}^{[0]}(\lambda;x,t)\mapsto \mathbf{M}^{[1]}(\lambda;x,t)$ preserves the normalization: $\mathbf{M}^{[1]}(\lambda;x,t) \to \mathbb{I}$ as $\lambda\to \infty$. 

As was summarized in \cite[Theorem 3.5]{BilmanM2019}, the function $q_1(x,t)$ defined by \eqref{q1-def} is recovered from
\begin{equation}
q_1(x,t) = 2\ii \lim_{\lambda\to \infty}\lambda M^{[1]}_{12}(\lambda;x,t)
\end{equation}
and it is a global solution of the focusing NLS equation \eqref{nls}. A virtue of this framework is that since $\mathbf{U}^{[1]}(\lambda;x,t)$ is holomorphic in $\lambda$ for $\lambda \in \mathbb{C}\setminus \Sigma^{\sharp}$, the Darboux transformation $\mathbf{M}^{[0]}\mapsto \mathbf{M}^{[1]}$ can be applied in an iterative manner arbitrarily many times even using the same special point $\lambda=\xi$. Thus, we define for $n=1,2,\ldots$,
\begin{equation}
\mathbf{U}^{[n]}(\lambda;x,t) := 
\begin{cases}
\mathbf{G}^{[n-1]}(\lambda;x,t)\mathbf{U}^{[n-1]}(\lambda;x,t),&\quad \lambda\in D_\infty^+ \cup D_\infty^-,\\
\mathbf{G}^{[n-1]}(\lambda;x,t)\mathbf{U}^{[n-1]}(\lambda;x,t)\mathbf{G}^{[n-1]}(\lambda;x_0,t_0)^{-1},&\quad \lambda\in D_\circ,\\
\end{cases}
\end{equation}
where $\mathbf{G}^{[n-1]}(\lambda;x,t)$ is the Darboux transformation matrix constructed from the simultaneous solution matrix $\mathbf{U}^{[n-1]}(\lambda;x,t)$ to the system of Lax pair equations \eqref{Lax-pair} with the potential $q=q_{n-1}(x,t)$, having the same form as \eqref{G-0} in which the seed solution is $\mathbf{U}^{[n-1]}(\lambda;x,t)\mathbf{c}$, where the column vector $\mathbf{c}$ is chosen to be \emph{the same} vector (with non-zero elements) at each stage of the iteration, i.e., for each $n=1,2,\ldots$. It is easy to see that the related matrix function
\begin{equation}
\mathbf{M}^{[n]}(\lambda;x,t)  := \mathbf{U}^{[n]}(\lambda;x,t) \ee^{\ii \lambda(x+\lambda t)\sigma_3}
\label{M-n}
\end{equation}
satisfies the following jump condition on $\Sigma^\sharp$:
\begin{equation}
\mathbf{M}^{[n]}_+(\zeta;x,t)=\mathbf{M}^{[n]}_-(\zeta;x,t) \ee^{-\ii \zeta (x + \zeta t)\sigma_3} \mathbf{V}^{[n]}(\zeta) \ee^{\ii \zeta (x + \zeta t)\sigma_3},\qquad \zeta \in \Sigma^\sharp,
\end{equation}
where the ``core'' of the jump matrix is 
\begin{equation}
\mathbf{V}^{[n]}(\zeta):=
\begin{cases}
\mathbf{G}^{[n-1]}(\zeta;x_0,t_0)\cdots \mathbf{G}^{[0]}(\zeta;x_0,t_0) \mathbf{U}(\zeta;x_0,t_0),&\quad \zeta \in \Sigma_\circ^+,\\
\mathbf{U}(\zeta;x_0,t_0)^{-1}\mathbf{G}^{[0]}(\zeta;x_0,t_0)^{-1} \cdots \mathbf{G}^{[n-1]}(\zeta;x_0,t_0)^{-1},&\quad \zeta \in \Sigma_\circ^-,\\
\mathbf{V}^{\mathbb{R}}(\zeta),&\quad \zeta \in \Sigma_{\mathrm{L}}\cup \Sigma_{\mathrm{R}}.
\end{cases}
\end{equation}
Next, we note the remarkable simplification \cite[Proposition 3.10]{BilmanM2019} that occurs in the products of $\mathbf{G}^{[k]}(\zeta;x_0,t_0)$ that appear in the jump matrix supported on the arcs $\Sigma_\circ^{\pm}$. Observe that for any integer $n\geq 1$ we have for $\lambda\in D_\circ$
\begin{equation}
\mathbf{U}^{[n]}(\lambda;x,t) = \mathbf{G}^{[n-1]}(\lambda;x,t)\cdots \mathbf{G}^{[0]}(\lambda;x,t) \mathbf{F}(\lambda;x,t) \mathbf{G}^{[0]}(\lambda;x,t)^{-1}\cdots \mathbf{G}^{[n-1]}(\lambda;x,t)^{-1},
\end{equation}
hence $\mathbf{U}^{[n]}(\xi;x_0,t_0) = \mathbb{I}$ for all $\lambda \in D_0$, implying $\frac{\dd}{\dd \lambda} \mathbf{U}^{[n]}(\lambda;x_0,t_0) \equiv \mathbf{0}$ for $\lambda\in D_0$, in particular, at $\lambda=\xi,\xi^*$. Thus, the building blocks of $\mathbf{G}^{[n]}(\lambda;x,t)$ given in \eqref{Y}--\eqref{N-w}
satisfy
\begin{equation}
\mathbf{s}(x_0,t_0) = \mathbf{c},\qquad N(x_0,t_0) = \|\mathbf{c}\|^2,\qquad w(x_0,t_0) = 0.
\label{N-w-0}
\end{equation}
Since these quantities are independent of $n$, it follows that
\begin{equation}
\mathbf{G}^{[n]}(\lambda;x_0,t_0) = \mathbf{G}^{[0]}(\lambda;x_0,t_0)
\end{equation}
for any integer $n\geq 0$ and for all $\lambda\in\mathbb{C}\setminus \{ \xi, \xi^* \}$. Thus, the matrix function $\mathbf{M}^{[n]}(\lambda;x,t)$ defined by \eqref{M-n}
solves the following Riemann-Hilbert problem.
\begin{rhp}\label{rhp:M-n}
Find a $2\times 2$ matrix valued function $\mathbf{M}^{[n]}(\lambda)=\mathbf{M}^{[n]}(\lambda;x,t)$ with the following properties.
\begin{itemize}
\item {\bf Analyticity:} $\mathbf{M}^{[n]}(\lambda)$ is analytic for $\lambda \in \mathbb{C}\setminus \Sigma^\sharp$.
\item {\bf Jump conditions:} $\mathbf{M}^{[n]}(\lambda)$ admits H\"older continuous boundary values on $\Sigma^\sharp$ for all exponents $\mu<\nu$ as $\lambda\to \zeta$ for $\zeta\in\Sigma^{\sharp}$, and these boundary values are related by the following jump condition
\begin{equation}
\mathbf{M}^{[n]}_+(\zeta)=\mathbf{M}^{[n]}_-(\zeta) \ee^{-\ii \zeta (x + \zeta t)\sigma_3} \mathbf{V}^{[n]}(\zeta) \ee^{\ii \zeta (x + \zeta t)\sigma_3},\qquad \zeta \in \Sigma,
\end{equation}
where the ``core'' of the jump matrix is 
\begin{equation}
\mathbf{V}^{[n]}(\zeta)=
\begin{cases}
\mathbf{G}^{[0]}(\zeta;x_0,t_0)^n \mathbf{U}(\zeta;x_0,t_0),&\quad \zeta \in \Sigma_\circ^+,\\
\mathbf{U}(\zeta;x_0,t_0)^{-1}\mathbf{G}^{[0]}(\zeta;x_0,t_0)^{-n},&\quad \zeta \in \Sigma_\circ^-,\\
\mathbf{V}^{\mathbb{R}}(\zeta),&\quad \zeta \in \Sigma_{\mathrm{L}}\cup \Sigma_{\mathrm{R}},
\end{cases}
\label{V-n}
\end{equation}
\item {\bf Normalization:} $\mathbf{M}^{[n]}(\lambda) \to \mathbb{I}$ as $\lambda\to\infty$ in $\mathbb{C}\setminus \Sigma$.
\end{itemize}
\end{rhp}
Note that the jump matrix supported on $\Sigma_\mathrm{L}\cup \Sigma_\mathrm{R} \subset \mathbb{R}$ remains unmodified (see \eqref{V-1}, or more generally, \eqref{V-n}) upon application of the Darboux transformation is because the action of the transformation
$\mathbf{M}^{[n-1]} \mapsto \mathbf{M}^{[n]}$ for $|\lambda|>r$ is nothing but multiplying $\mathbf{M}^{[n-1]}(\lambda;x,t)$ on the left by $\mathbf{G}^{[n-1]}(\lambda;x,t)$ which is analytic in $\lambda$ for $|\lambda|>r$. As was shown in \cite{BilmanM2019}, \rhref{rhp:M-n} is uniquely solvable and 
\begin{equation}
q_n(x,t):=2\ii \lim_{\lambda\to\infty} \lambda M_{12}^{[n]}(\lambda;x,t)
\label{q-n-recovery}
\end{equation}
is a global solution of the focusing NLS equation \eqref{nls}.

We now use the Riemann-Hilbert matrix $\mathbf{M}^{[0]}(\lambda;x,t)$ defined by \eqref{U-0} and \eqref{M-0} (solution of \rhref{rhp:M-n} with $n=0$) for the background field $q_0(x,t)$ as a global parametrix and introduce the perturbative matrix function
\begin{equation}
\mathbf{P}^{[n]}(\lambda;x,t) := \mathbf{M}^{[n]}(\lambda;x,t) \mathbf{M}^{[0]}(\lambda;x,t)^{-1}.
\label{P-n}
\end{equation}
First, note from \rhref{rhp:background} and \rhref{rhp:M-n} that $\mathbf{M}^{[0]}(\lambda;x,t)$ and $\mathbf{M}^{[n]}(\lambda;x,t)$ satisfy exactly the same jump condition on $\Sigma_\mathrm{L}\cup \Sigma_\mathrm{R}$. Thus, $\mathbf{P}^{[n]}(\lambda;x,t)$ extends analytically across $\Sigma_\mathrm{L}\cup \Sigma_\mathrm{R}$, hence the jump discontinuity on these arcs is removed. Next, for $\zeta\in\Sigma_\circ^+$ we have
\begin{equation}
\begin{aligned}
\mathbf{P}^{[n]}_+(\zeta;x,t) &= \mathbf{M}^{[n]}_+(\zeta;x,t) \left( \mathbf{U}(\zeta;x,t) \ee^{\ii \zeta(x+\zeta t)\sigma_3}\right)^{-1}\\
&= \mathbf{M}^{[n]}_-(\zeta;x,t)  \ee^{-\ii \zeta(x+\zeta t)\sigma_3} \mathbf{G}^{[0]}(\zeta;x_0,t_0)^n \mathbf{U}(\zeta;x_0,t_0) \ee^{\ii \zeta(x+\zeta t)\sigma_3}\left( \mathbf{U}(\zeta;x,t) \ee^{\ii \zeta(x+\zeta t)\sigma_3}\right)^{-1}\\
&= \mathbf{M}^{[n]}_-(\zeta;x,t) \mathbf{M}_-^{[0]}(\zeta;x,t)^{-1} \left(\mathbf{F}(\zeta;x,t)\ee^{\ii \zeta(x+\zeta t)\sigma_3} \right) \ee^{-\ii \zeta(x+\zeta t)\sigma_3} \mathbf{G}^{[0]}(\zeta;x_0,t_0)^n \mathbf{F}(\zeta;x,t)^{-1} \\
&= \mathbf{P}^{[n]}_-(\zeta;x,t)\mathbf{F}(\zeta;x,t) \mathbf{G}^{[0]}(\zeta;x_0,t_0)^n \mathbf{F}(\zeta;x,t)^{-1}.
\end{aligned}
\label{P-jump-plus}
\end{equation}
Similarly, for $\zeta\in\Sigma_\circ^-$,
\begin{equation}
\begin{aligned}
\mathbf{P}^{[n]}_+(\zeta;x,t)&=\mathbf{M}^{[n]}_+(\zeta;x,t)\left(\mathbf{F}(\zeta;x,t)\ee^{\ii\zeta(x+\zeta t)\sigma_3}\right)^{-1}\\
&=\mathbf{M}^{[n]}_-(\zeta;x,t)\ee^{-\ii\zeta(x+\zeta t)\sigma_3}\mathbf{U}(\zeta;x_0,t_0)^{-1}\mathbf{G}^{[0]}(\zeta;x_0,t_0)^{-n}\ee^{\ii\zeta(x+\zeta t)\sigma_3}\left(\mathbf{F}(\zeta;x,t)\ee^{\ii\zeta(x+\zeta t)\sigma_3}\right)^{-1}\\
&=\mathbf{M}_-^{[n]}(\zeta;x,t)\mathbf{M}_-^{[0]}(\zeta;x,t)^{-1}\mathbf{F}(\zeta;x,t)\mathbf{G}^{[0]}(\zeta;x_0,t_0)^{-n}\mathbf{F}(\zeta;x,t)^{-1}\\
&=\mathbf{P}^{[n]}_-(\zeta;x,t)\mathbf{F}(\zeta;x,t)\mathbf{G}^{[0]}(\zeta;x_0,t_0)^{-n}\mathbf{F}(\zeta;x,t)^{-1}.
\end{aligned}
\label{P-jump-minus}
\end{equation}
It is also immediate from the definition \eqref{P-n} that $\mathbf{P}^{[n]}(\lambda;x,t)\to \mathbb{I}$ as $\lambda\to \infty$. 
Since the jump matrices computed in \eqref{P-jump-plus} and \eqref{P-jump-minus} are inverses of each other, we reverse the orientation of the semi-circle $\Sigma_\circ^-$ so that the circle $\Sigma_\circ$ is oriented clockwise. Then $\mathbf{P}^{[n]}(\lambda;x,t)$ solves the following Riemann-Hilbert problem:

\begin{rhp}\label{rhp:perturbative}
Let $(x,t)\in\mathbb{R}^2$ be arbitrary parameters, and let $n\in \mathbb{Z}_{\geq 0}$. Find a $2\times 2$ matrix function $\mathbf{P}^{[n]}(\lambda;x,t)$ that has the following properties:
  \begin{itemize}
    \item \textbf{Analyticity:} $\mathbf{P}^{[n]}(\lambda;x,t)$ is analytic for $\lambda \in \mathbb{C}\setminus \Sigma_\circ$.
    \item \textbf{Jump Condition:} $\mathbf{P}^{[n]}(\lambda;x,t)$ takes continuous boundary values on $\Sigma_\circ$ denoted by $\mathbf{P}^{[n]}_\pm(\zeta;x,t)$, $\zeta\in\Sigma_\circ$, and they are related by the following jump condition:
\begin{equation}
\mathbf{P}^{[n]}_+(\zeta;x,t) = \mathbf{P}^{[n]}_-(\zeta;x,t)
\mathbf{F}(\zeta;x,t)\mathbf{G}^{[0]}(\zeta;x_0,t_0)^n\mathbf{F}(\zeta;x,t)^{-1},\quad \zeta \in \Sigma_\circ,
\label{P-jump-circle}
\end{equation}
in which the circle $\Sigma_\circ$ is now given clockwise orientation (for both semicircles).
\item \textbf{Normalization:} $ \mathbf{P}^{[n]}(\lambda;x,t) \to \mathbb{I}$ as $\lambda\to\infty$.
\end{itemize}
\end{rhp}
It follows from the formul\ae{} \eqref{q-0-recovery} and \eqref{q-n-recovery} that
\begin{equation}
q_n(x,t) = q_0(x,t) + 2\ii \lim_{\lambda\to \infty}\lambda P^{[n]}_{12}(\lambda;x,t).
\label{eq:qn-q0}
\end{equation}
Incidentally, the value $\mathbf{G}^{[0]}(\lambda;x_0,t_0)$ at $(x,t)=(x_0,t_0)$ of the matrix $\mathbf{G}^{[0]}(\lambda;x,t)$ constructed here coincides with the value of its analogue used in \cite{BilmanB2019} to construct multiple-pole solitons on a zero background. The reason is because in either case we have used the simultaneous solution $\mathbf{F}(\lambda;x,t)$ of the Lax pair equations that is normalized to be the identity matrix at $(x,t)=(x_0,t_0)$ in a disk large enough to contain $\lambda=\xi$ ($x_0$ was taken to be $0$ and $t_0=0$ in \cite{BilmanB2019}, see Remark~\ref{rem:point-L} below), thus the presence of an underlying non-trivial background field $q_0$ is invisible to $\mathbf{F}(\lambda;x,t)$ whenever $(x,t)=(x_0,t_0)$, and hence to $\mathbf{G}^{[0]}(\lambda;x,t)$ whenever $(x,t)=(x_0,t_0)$. By direct calculation using \eqref{N-w-0} we have
\begin{equation}
\mathbf{G}^{[0]}(\lambda;x_0,t_0) = \mathbb{I} + \frac{2\ii \beta}{\|\mathbf{c}\|^2} 
\begin{bmatrix} |c_2|^2 & - c_1 c_2^* \\ -c_1^* c_2 & |c_1|^2 \end{bmatrix}\frac{1}{\lambda-\xi} -  \frac{2\ii \beta}{\|\mathbf{c}\|^2} 
\begin{bmatrix} |c_1|^2 &  c_1 c_2^* \\ c_1^* c_2 & |c_2|^2 \end{bmatrix}\frac{1}{\lambda-\xi^*},
\end{equation}
and $\mathbf{G}^{[0]}(\lambda;x_0,t_0)$ can be diagonalized as
\begin{equation}
\mathbf{G}^{[0]}(\lambda;x_0,t_0) =\mathbf{Q}\left(  \frac{\lambda - \xi}{\lambda-\xi^*} \right)^{\sigma_3}\mathbf{Q}^{-1},\qquad \mathbf{Q}:=\frac{1}{\|\mathbf{c}\|}\begin{bmatrix} 
c_1 & -c_2^* \\ c_2 & c_1^* \end{bmatrix}.
\label{eq:Gnaught}
\end{equation}
{
\begin{remark} 
\label{rem:c1-c2-0}
Although $\mathbf{c}=[1\colon 0]$ and $\mathbf{c}=[0\colon 1]$ are elements of $\mathbb{CP}^1$ that are perfectly fine to use as data in the construction of $\mathbf{G}(\lambda;x,t)$, it is easy to see that 
\begin{equation}
\mathbf{G}^{[0]}(\zeta;x_0,t_0)^n = 
\begin{cases}
\displaystyle\left(\frac{\zeta - \xi}{\zeta-\xi^*}\right)^{-n\sigma_3},&\quad\text{if}~c_1=0,\\
\displaystyle\left(\frac{\zeta - \xi}{\zeta-\xi^*}\right)^{n\sigma_3},&\quad\text{if}~c_2=0.\\
\end{cases}
\end{equation}
This has two consequences. First, if the background field $q_0(x,t)$ is chosen such that the associated simultaneous solution matrix $\mathbf{F}(\lambda;x,t)$ (see \eqref{F-normalization}--\eqref{F-to-U}) is diagonal, then the jump matrix in \eqref{P-jump-circle} becomes diagonal and independent of $(x,t)$ whenever $c_1=0$ or $c_2=0$. Then \rhref{rhp:perturbative} can be solved explicitly and it follows that $q_n(x,t)=q_0(x,t)$ for all $(x,t)\in \mathbb{R}^2$ and for each $n=1,2,\ldots$, i.e., the underlying B\"acklund transformation degenerates. This happens, for instance, when $q_0\equiv 0$, in which case $\mathbf{F}(\zeta;x,t)= \ee^{-\ii(\zeta(x-x_0)+\zeta^2(t-t_0))\sigma_3}$. This is why the choice $c_1c_2\neq 0$ was employed in \cite{BilmanB2019} for solitons on the zero background. 
Second, regardless of what the background field is (i.e., regardless of what the matrix $\mathbf{F}$ in \eqref{P-jump-circle} is), the jump matrix in \eqref{P-jump-circle} becomes diagonal if evaluated at $(x,t)=(x_0,t_0)$ whenever $c_1=0$ or $c_2=0$ due to \eqref{F-normalization}. This implies that $q_n(x_0,t_0) = q_0(x_0,t_0)$ for each $n=1,2,\ldots$, but the equality need not hold for other values of $(x,t)$. The simplest example of this scenario is when $q_0 \equiv \ee^{\ii t}$, in which case $\mathbf{F}(\zeta;x,t)$ is the full matrix given in \cite[Eqn.\@ (16)]{BilmanLM2020} and
$q_n(x,t) \neq q_0(x_0,t_0)$ for general values of $(x,t)$ even when $c_1=0$ or $c_2=0$ in this case. Indeed, plotting $|q_1(x,t)|$ when $q_0\equiv \ee^{\ii t}$ shows that it is a localized bump with a peak attained at a location other than $(x_0,t_0)$.
\end{remark}
}
\begin{proposition} 
\label{prop:qn-growth}
The amplitude $|q_n(x_0,t_0)|$ grows proportionally to $n$ as $n$ increases {if $c_1c_2 \neq 0$}. In particular,
\begin{equation}
q_n(x_0,t_0) = q_0(x_0, t_0) +8\beta \frac{ c_1 c_2^*}{\|\mathbf{c}\|^2}n.
\label{qn-amplitude}
\end{equation}
\end{proposition}
\begin{proof}
From the normalization \eqref{F-normalization} and the diagonalization \eqref{eq:Gnaught} the jump condition \eqref{P-jump-circle} becomes
\begin{equation}
\mathbf{P}^{[n]}_+(\zeta;x_0,t_0) = \mathbf{P}^{[n]}_-(\zeta;x_0,t_0)
\mathbf{Q}\left(  \frac{\lambda - \xi}{\lambda-\xi^*} \right)^{n \sigma_3}\mathbf{Q}^{-1}
,\quad \zeta \in \Sigma_\circ,
\end{equation}
Since the diagonalization \eqref{eq:Gnaught} via conjugation by the constant matrix $\mathbf{Q}$ preserves the normalization as $\lambda\to\infty$, it is easy to see that
\begin{equation}
\mathbf{P}^{[n]}(\lambda;x_0,t_0) = 
\begin{cases}
\mathbb{I},&\quad |\lambda|<r\\
\mathbf{Q}\left(  \frac{\lambda - \xi}{\lambda-\xi^*} \right)^{n \sigma_3}\mathbf{Q}^{-1},&\quad |\lambda|>r.
\end{cases} 
\end{equation}
Therefore,
$\mathbf{P}^{[n]}(\lambda;x_0,t_0) = \mathbb{I} - 2\ii \beta n [\mathbf{Q}\sigma_3 \mathbf{Q}^{-1}]\lambda^{-1} + O(\lambda^{-2})$ as $\lambda\to \infty$,
hence
\begin{equation}
P^{[n]}_{12}(\lambda;x_0,t_0) = -4\ii\beta  \frac{ c_1 c_2^*}{\|\mathbf{c}\|^2}n \lambda^{-1} + O(\lambda^{-2}),\quad \lambda\to \infty.
\end{equation}
From this and \eqref{eq:qn-q0} the result follows.
\end{proof}
\begin{remark}
Note that Proposition~\ref{prop:qn-growth} does not necessarily imply that $|q_n(x_0,t_0)|$ is a maximum value of $|q_n(x,t)|$ in $(x,t)\in\mathbb{R}^2$ for given $n\in\mathbb{Z}_{>0}$. 
{Nevertheless, it provides a generalization of the result for the peak amplitude of fundamental rogue waves given in \cite[Proposition 2]{BilmanLM2020} to the superposition of coherent structures on a suitably arbitrary background field $q_0(x,t)$; it also generalizes the same result obtained in \cite[Theorem 1]{WangYWH2017} by algebraic methods. The result \eqref{qn-amplitude} also gives the amplitude of the multiple-pole solitons of order $2n$ characterized by \cite[Riemann-Hilbert Problem 1 and Eqn.\@ (1.5)]{BilmanB2019} at the point $(x_0,t_0)$ (taken to be $(0,0)$ in that reference) if the background $q_0$ is taken to be the zero background $q_0\equiv 0$, in which case $\mathbf{F}(\lambda;x,t) = \ee^{-\ii(\lambda(x-x_0) + \lambda^2(t-t_0))\sigma_3}$.
The point of Proposition~\ref{prop:qn-growth} is that the amplitude of the deviation $|q_n(x,t) - q_0(x,t)|$ from the background field at $(x,t)=(x_0,t_0)$ grows proportionally to $n$ as $n$ increases although $(x_0,t_0)$ may not be the location of the peak of the wave profile.
}
\end{remark}
\begin{remark}
The robust version of the inverse-scattering transform was introduced in \cite{BilmanM2019} in the context of an initial-value problem for the nonlinear Schr\"odinger equation \eqref{nls} posed for $x\in\mathbb{R}$ and $t>0$. Thus, the jump matrix supported on $\Sigma_\circ$ was given in terms of scattering data computed at $t=0$ from a given suitable initial condition. Accordingly, the analogues of \eqref{V-0} in \cite{BilmanM2019} have $t_0=0$ (and $x_0=L$, arbitrary). The iterated Darboux transformation scheme introduced in this context in \cite{BilmanM2019} was also devised to place the peaks of fundamental rogue waves at $(x,t)=(L,0)$, where $L$ was taken to be $0$ without loss of generality.
This choice was retained first in the study of the near-field behavior of the \emph{extreme superposition} of fundamental rogue waves by the authors in \cite{BilmanLM2020} and also in its adaptation to solitons in \cite{BilmanB2019}. 
In our work, the starting point is a global solution $q_0(x,t)$ of \eqref{nls} provided by the solution of \rhref{rhp:background} and consequently, we may use any point $(x_0,t_0)$ in the $(x,t)$-plane as a normalization point in the sense of \eqref{F-normalization}. As the reader shall see in Section~\ref{sec:conclusion}, this flexibility results in a further, two-fold universality.
\label{rem:point-L}
\end{remark}

\subsection{Extreme superposition and passage to the limit}
\label{s:extreme}
Observe that we may write
\begin{equation}
\mathbf{F}(\lambda;x,t) = \mathbf{U}(\lambda;x,t)  \mathbf{U}(\lambda;x_0,t_0)^{-1} = \mathbf{M}(\lambda;x,t)\ee^{- \ii \lambda((x -x_0) + \lambda (t-t_0))\sigma_3}\mathbf{M}(\lambda;x_0,t_0)^{-1},
\end{equation}
and the jump condition \eqref{P-jump-circle} can be expressed in terms of the matrix $\mathbf{M}$ associated with the background field $q_0$ as follows
\begin{multline}
\mathbf{P}^{[n]}_+(\zeta;x,t) = \mathbf{P}^{[n]}_-(\zeta;x,t)
 \mathbf{M}(\zeta;x,t)\ee^{-\ii \zeta((x - x_0) + \zeta (t-t_0))\sigma_3}\mathbf{M}(\zeta;x_0,t_0)^{-1}
 \\
  \times \mathbf{G}^{[0]}(\zeta;x_0,t_0)^n\mathbf{M}(\zeta;x_0,t_0)\ee^{\ii \zeta((x - x_0) + \zeta (t-t_0))\sigma_3}\mathbf{M}(\zeta;x,t)^{-1},\quad \zeta \in \Sigma_\circ.
\end{multline}
Now we introduce new variables as follows:
\begin{equation}
\Lambda:=n^{-1}\lambda,\quad X:=n(x-x_0),\quad T:=n^2(t-t_0).
\end{equation}
Since the radius of the circular contour $\Sigma_\circ$ can be taken to be arbitrarily large, we choose it to coincide with the unit circle in the $\Lambda$-plane.  Then defining also $\mathbf{R}^{[n]}(\Lambda;X,T):=\mathbf{P}^{[n]}(n\Lambda;x_0+n^{-1}X,t_0+n^{-2}T)$, we see that $\mathbf{R}^{[n]}(\Lambda;X,T)$ is analytic in $\Lambda$ for $|\Lambda|\neq 1$, and with the unit circle taken with clockwise orientation, the jump condition for $\mathbf{R}^{[n]}(\Lambda;X,T)$ reads
\begin{multline}
\mathbf{R}_+^{[n]}(\Xi;X,T)=\mathbf{R}_-(\Xi;X,T)\mathbf{M}(n\Xi;x_0+n^{-1}X,t_0+n^{-2}T)\ee^{-\ii\Xi(X+\Xi T)\sigma_3}\mathbf{M}(n\Xi;x_0,t_0)^{-1}
\\
{}\times\mathbf{G}^{[0]}(n\Xi;x_0,t_0)^n\mathbf{M}(n\Xi;x_0,t_0)\ee^{\ii\Xi(X+\Xi T)\sigma_3}\mathbf{M}(n\Xi;x_0+n^{-1}X,t_0+n^{-2}T)^{-1},\quad |\Xi|=1.
\end{multline}
According to Proposition~\ref{prop:independence}, $\mathbf{M}(n\Xi;x_0,t_0)\to \mathbb{I}$ as $n\to\infty$ uniformly for $|\Xi|=1$.  Likewise, writing
\begin{equation}
\mathbf{M}(n\Xi;x_0+n^{-1}X,t_0+n^{-2}T)=[\mathbf{M}(n\Xi;x_0+n^{-1}X,t_0+n^{-2}T)-\mathbf{M}(n\Xi;x_0,t_0)]+\mathbf{M}(n\Xi;x_0,t_0),
\end{equation}
the first term tends to zero as $n\to\infty$ for $(X,T)$ in any bounded set according to Proposition~\ref{prop:continuity}, and the second term tends to $\mathbb{I}$ as $n\to\infty$ uniformly for $|\Xi|=1$ by Proposition~\ref{prop:independence}.  Finally, using \eqref{eq:Gnaught} shows that
\begin{equation}
\mathbf{G}^{[0]}(n\Xi;x_0,t_0)^n = \mathbf{Q}\left(\frac{n\Xi-\xi}{n\Xi-\xi^*}\right)^{n\sigma_3}\mathbf{Q}^{-1} \to\mathbf{Q}\ee^{-2\ii\beta\Xi^{-1}\sigma_3}\mathbf{Q}^{-1},\quad n\to\infty
\end{equation}
holds uniformly for $|\Xi|=1$, where we recall the notation $\beta:=\mathrm{Im}(\xi)>0$.  The limiting jump condition has the Schwarz symmetry necessary for the vanishing lemma to apply, and the first moment of the solution of the limiting problem is easily seen to satisfy the nonlinear Schr\"odinger equation in the variables $(X,T)$ by a dressing argument; hence by $L^2(\Sigma_\circ)$ small-norm theory for Riemann-Hilbert problems \cite[Appendix B]{KamvissisMM2003} the following theorem is proved.
\begin{theorem}
\label{thm:main}
Let $q_0(x,t)$ be an arbitrary background potential solving the focusing nonlinear Schr\"odinger equation in the form \eqref{nls} and obtained from\footnote{Or obtained from an arbitrary solution of \eqref{nls} satisfying nonzero boundary conditions as described in Remark~\ref{rem:nonzero-background}.} \rhref{rhp:background}, and let $q_n(x,t)$ denote the result of the $n$-fold Darboux transformation given by \eqref{q-n-recovery} or \eqref{eq:qn-q0}, and defined in terms of parameters $(x_0,t_0)\in\mathbb{R}^2$, $\xi\in\mathbb{C}$ with $\mathrm{Im}(\xi)=\beta>0$, and  $\mathbf{c}\in\mathbb{CP}^1$ (see \eqref{eq:Gnaught}).  Then uniformly for $(X,T)$ in any bounded set,
\begin{equation}
\label{Q-def}
\lim_{n\to\infty}n^{-1}q_n(x_0+n^{-1}X,t_0+n^{-2}T) = Q(X,T):=2\ii\lim_{\Lambda\to\infty}\Lambda R_{12}(\Lambda;X,T),
\end{equation}
where $\mathbf{R}(\Lambda;X,T)$ is analytic for $|\Lambda|\neq 1$, $\mathbf{R}(\Lambda;X,T)\to\mathbb{I}$ as $\Lambda\to\infty$, and where $\mathbf{R}(\Lambda;X,T)$ takes analytic boundary values on the unit circle with clockwise orientation that are related by the jump condition
\begin{equation}
\label{R-jump}
\mathbf{R}_+(\Xi;X,T)=\mathbf{R}_-(\Xi;X,T)\ee^{-\ii\Xi(X+\Xi T)\sigma_3}\mathbf{Q}\ee^{-2\ii\beta\Xi^{-1}\sigma_3}\mathbf{Q}^{-1}\ee^{\ii\Xi(X+\Xi T)\sigma_3},\quad |\Xi|=1.
\end{equation}
The matrix function $\mathbf{R}(\Lambda;X,T)$ is uniquely determined by these conditions, and $Q(X,T)$ is a solution of the focusing nonlinear Schr\"odinger equation in the form
\begin{equation}
\ii Q_T + \frac{1}{2}Q_{XX}+|Q|^2Q=0,\quad (X,T)\in\mathbb{R}^2.
\end{equation}
\end{theorem}  
{
The function $Q(X,T)$ is a generalization of the \emph{rogue wave of infinite order} \cite{BilmanLM2020}.  It depends parametrically on $\mathbf{c}\in\mathbb{CP}^1$ (via $\mathbf{Q}$) and $\beta>0$.  The dependence on $\beta$ can be easily scaled out by the similarity transformation $Q(X,T)\mapsto \beta^{-1}Q(\beta^{-1}X,\beta^{-2}T)$, however the dependence on $\mathbf{c}\in\mathbb{CP}^1$ is nontrivial.  In general, $Q(X,T)$ is a highly transcendental function, satisfying differential equations in the Painlev\'e-III hierarchy in $X$ and $T$ independently (see \cite{BilmanB2019,BilmanLM2020}).  However when $c_1c_2=0$ the situation simplifies.
Indeed, we note that the jump matrix in \eqref{R-jump} equals $\ee^{2\ii\beta\Xi^{-1}\sigma_3}$ if $c_1=0$ and equals $\ee^{-2\ii\beta\Xi^{-1}\sigma_3}$ if $c_2=0$. In both cases $\mathbf{R}(\Lambda;X,T)$ can be determined explicitly and is equal to the corresponding diagonal jump matrix for $|
\Lambda|>1$. It then follows from \eqref{Q-def} that $Q(X,T)\equiv 0$. There could be different reasons for this degeneration in Theorem~\ref{thm:main} in view of Remark~\ref{rem:c1-c2-0}. First, it could be that $q_n(x,t)= q_0(x,t)$ for all $(x,t)\in\mathbb{R}^2$ for each $n=1,2,\ldots$, when the construction is performed with $c_1=0$ or $c_2=0$, just like the example of solitons on the zero background given in Remark~\ref{rem:c1-c2-0}. More interestingly, it could be that the solution $q_n(x,t)$ is nontrivial and $|q_n(x,t)|$ grows proportionally to $n$ as $n\to +\infty$ at a (possibly $n$-dependent) location whose (possibly $n$-dependent) distance to $(x_0,t_0)$ fails to be captured by the rate at which we are zooming in at the point $(x_0,t_0)$ in \eqref{Q-def}. Alternatively, it could be that the aforementioned peak amplitude grows at a rate $o(n)$ as $n\to+\infty$ even if the location is captured by the rate at which we are zooming in at the point $(x_0,t_0)$.
}
{
\subsection{Conclusion}
\label{sec:conclusion}
We established that generation of rogue waves of infinite order exhibits a universal character in two ways. They can be generated on given suitably arbitrary background solution $q_0(x,t)$ of the focusing NLS equation \eqref{nls} and at an arbitrary point $(x_0,t_0)$ of the spacetime domain of that background field. This generalizes the emergence of the special solution $Q=\Psi^{\pm}(X,T)$ in the near-field behavior of fundamental rogue waves on a Stokes wave background $q_0\equiv \ee^{\ii t}$ identified in \cite{BilmanLM2020} and in the near-field behavior of solitons on the zero background in the limit of large order identified in \cite{BilmanB2019} to arbitrary background fields and locations. Our result also complements the universality of this near-field behavior the authors established recently in \cite{BilmanM2021} for background fields described by \cite[Riemann-Hilbert Problem 2]{BilmanM2021} with $0<M<\frac{1}{2}$, $M\neq \frac{1}{4}$, in the notation of that reference.   We emphasize that, while we aimed to capture many different types of background solutions $q_0(x,t)$ to \eqref{nls} in the setting of \rhref{rhp:background}, both the iterated B\"acklund transformations and the resulting convergence obtained in Theorem~\ref{thm:main} are of a local character with respect to $(x,t)$ near $(x_0,t_0)$ and as such it is likely that the theorem holds true for a much broader class of  background solutions defined in a neighborhood of the chosen point $(x_0,t_0)$.
}

\end{document}